\documentclass[sigconf,nonacm]{acmart}

\newcommand{\MPC}{\textsf{MPC}}
\newcommand{\AMPC}{\textsf{AMPC}}
\newcommand{\PRAM}{\textsf{PRAM}}

\newcommand{\nn}{\mathsf{N}}
\newcommand{\memory}{\mathsf{M}}

\newcommand{\machines}{\mathsf{P}}

\newcommand {\ignore} [1] {}

\usepackage{graphicx}
\usepackage[ruled,vlined,linesnumbered]{algorithm2e}

\usepackage{graphics}
\usepackage{xcolor}
\usepackage{amsfonts,amsxtra}
\usepackage{amsmath,amsthm}

\usepackage{float}
\usepackage{pictex}

\usepackage{caption}
\usepackage{subcaption}
\usepackage{tikz}
\usetikzlibrary{arrows,shapes,decorations,patterns,automata,backgrounds,petri,calc}
\usepackage[latin1]{inputenc}

\newcommand{\bigO}{\mathrm{O}}

\floatstyle{ruled}
\newfloat{Figure}{t}{fig}

\newcommand{\mc}{Min Cut}
\newcommand{\mkc}{Min $k$-Cut}
\newcommand{\ceil}[1]{\lceil #1\rceil}
\newcommand{\bag}{\mathsf{bag}}
\newcommand{\bagLeader}{\mathsf{bagLeader}}
\newcommand{\ldrtime}{\mathsf{ldr}\_\mathsf{time}}
\newcommand{\nbag}{\mathsf{nbr}\_\mathsf{bag}}
\newcommand{\mw}{\mathsf{mw}}

\newtheorem{lemma}{Lemma}
\newtheorem{theorem}{Theorem}
\newtheorem{corollary}{Corollary}

\newtheorem{definition}{Definition}
\newtheorem{observation}{Observation}

\newcommand{\thmmincut}{There is an $O(\log\log{n})$-round AMPC algorithm that uses $\widetilde{O}(n + m)$ total memory and $\widetilde{O}(n^\epsilon)$ memory per machine which finds a $(2+\epsilon)$-approximation of \mc{} with high probability.}

\newcommand{\thmvazirani}{Algorithm \textsf{APX-SPLIT} is an $(4 + \epsilon)$ approximation of the \mkc{}. Furthermore, it can be implemented in the AMPC model with $\bigO(n^{\epsilon})$ memory per machine in $\bigO(k \log\log{n})$ rounds and $\bigO(m)$ total memory.}

\newtheorem{innercustomthm}{Theorem}
\newenvironment{customthm}[1]
  {\renewcommand\theinnercustomthm{#1}\innercustomthm}
  {\endinnercustomthm}

\begin{document}

\title{Adaptive Massively Parallel Algorithms for Cut Problems}

\author{MohammadTaghi Hajiaghayi}
\affiliation{%
  \institution{University of Maryland}
  \city{College Park, Maryland}
  \country{USA}
}
\email{hajiagha@cs.umd.edu}

\author{Marina Knittel}
\affiliation{%
  \institution{University of Maryland}
  \city{College Park, Maryland}
  \country{USA}
}
\email{mknittel@cs.umd.edu}

\author{Jan Olkowski}
\affiliation{%
  \institution{University of Maryland}
  \city{College Park, Maryland}
  \country{USA}
}
\email{jan.olkowski@gmail.com}

\author{Hamed Saleh}
\affiliation{%
  \institution{University of Maryland}
  \city{College Park, Maryland}
  \country{USA}
}
\email{hameelas@gmail.com}

\renewcommand{\shortauthors}{Hajiaghayi, Knittel, Olkowski, and Saleh}

\begin{abstract}
We study the \emph{Weighted \mc{}} problem in the \emph{Adaptive Massively Parallel Computation} ($\AMPC$) model. 
In 2019, Behnezhad et al.~\cite{behnezhad2019massively} introduced the $\AMPC$ model as an extension of the \emph{Massively Parallel Computation} ($\MPC$) model.
In the past decade, research on highly scalable algorithms has had significant impact on many massive systems. 
The $\MPC$ model, introduced in 2010 by Karloff et al.~\cite{karloff2010a}, which is an abstraction of famous practical frameworks such as MapReduce, Hadoop, Flume, and Spark, has been at the forefront of this research. 
While great strides have been taken to create highly efficient MPC algorithms for a range of problems, recent progress has been limited by the \emph{1-vs-2 Cycle Conjecture}~\cite{yaroslavtsev2018massively}, which postulates that the simple problem of distinguishing between one and two cycles requires $\Omega(\log n)$ $\MPC$ rounds. 
In the $\AMPC$ model, each machine has adaptive read access to a distributed hash table even when communication is restricted (i.e., in the middle of a \textit{round}). While remaining practical~\cite{behnezhad2020parallel}, this gives algorithms the power to bypass limitations like the 1-vs-2 Cycle Conjecture. 

We give the first sublogarithmic $\AMPC$ algorithm, requiring $O(\log\log n)$ rounds, for $(2+\epsilon)$-approximate weighted \mc{}. Our algorithm is inspired by the divide and conquer approach of Ghaffari and Nowicki~\cite{ghaffari2020massively}, which solves the $(2+\epsilon)$-approximate weighted \mc{} problem in $O(\log n\log\log n)$ rounds of $\MPC$ using the classic result of Karger and Stein~\cite{karger1996new}. Our work is \emph{fully-scalable} in the sense that the local memory of each machine is $O(n^\epsilon)$ for any constant $0 < \epsilon < 1$. There are no $o(\log n)$-round $\MPC$ algorithms for \mc{} in this memory regime assuming the 1-vs-2 Cycle Conjecture holds.  The \emph{exponential speedup} in the AMPC runtime is the result of \emph{decoupling} the different layers of the divide and conquer algorithm and solving all layers in $O(1)$ rounds in parallel.
Finally, we extend our approach: we present an $O(k\log\log n)$-round $\AMPC$ algorithms for approximating the \mkc{} problem with a $4+\epsilon$ approximation factor.



\end{abstract}

\maketitle

\section{Introduction}

\emph{Massively Parallel Computation} ($\MPC$) -- introduced by Karloff et al.~\cite{karloff2010a} in 2010 -- is an abstract model that captures the capabilities of the modern parallel/distributed frameworks widely used in practice such as  MapReduce~\cite{dean2008mapreduce}, Hadoop~\cite{hadoop}, Flume~\cite{chambers2010flumejava}, and Spark~\cite{zaharia2016apache}. $\MPC$ has been at the forefront of the research on parallel algorithms in the past decade, and it is now known as the de facto standard computation model for the analysis of parallel algorithms. 

In this paper, we focus on sublogarithmic-round algorithms for the \mc{} problem in the \emph{Adaptive Massively Parallel Computation} ($\AMPC$) model, which is a recent extension of $\MPC$. In both $\MPC$ and $\AMPC$, the input data is far larger than the memory of a single machine, and thus an input of size $\nn$ is initially distributed across a collection of $\machines$ machines. In the $\MPC$ model, the algorithm executes in several synchronous rounds, in which each machine executes local computations isolated from other machines, and the machines can only communicate at the end of a round. The total size of incoming/outgoing messages for each machine is  limited by local memory constraints. We are interested in \emph{fully-scalable} algorithms in which every machine is allocated a local memory of size $O(\nn^\epsilon)$ for any constant $0 < \epsilon < 1$. Moreover, we can often improve the round complexity\footnote{The number of rounds is a main complexity of interest since in practice the bottleneck is often the communication phase.} of the massively parallel algorithms by allowing a super-linear total memory $O(\nn^{1+\epsilon})$, for example, the filtering technique of Lattenzi et al.~\cite{lattanzi2011filtering} in $\MPC$ or the maximal matching algorithm of Behnezhad et al.~\cite{behnezhad2020parallel} in $\AMPC$. So we are primarily interested in algorithms with $\widetilde{O}(\nn)$ total memory, and therefore we assume there are $\machines = \widetilde{O}(\nn^{1-\epsilon})$ machines.
\footnote{Where $\widetilde{O}$ hides polylogarithmic factors, i.e., $\widetilde{O}(f(n)) = O(f(n)poly\log(n))$.} 

Recent developments in the hardware infrastructure and new technologies such as RDMA~\cite{dragojevic2017rdma}, eRPC, and Farm~\cite{dragojevic2014farm} allow for high-throughput, low-latency communication among machines in data centers, such that remote volatile memory accesses are becoming faster than accessing local persistent storage. The concept of a shared remote memory is in particular useful when machines need to query data adaptively -- i.e., deciding what to query next based on the previously queried data -- which requires a communication round per query in the MPC model. Behnezhad et al.~\cite{behnezhad2019massively} incorporates this RDMA-like paradigm of remote memory access into the MPC model and introduces AMPC. In the new model, the machines can \textit{adaptively} query from a \emph{distributed hash table}, or a shared read-only memory, during each round. Machines are only allowed to write to shared memory at the end of each round. There is also empirical evidence that AMPC algorithms for several problems -- including maximal independent set, maximal matching, and connectivity -- obtain significant speedups in running time compared to state-of-the-art MPC algorithms~\cite{behnezhad2020parallel}. This fact, which stems from the meaningful drop in the number of communication rounds, verifies the practical power of the AMPC model.


In this paper, we provide the first AMPC-specific algorithms for the \mc{} problem. The \mc{} of a given graph $G=(V,E)$ is the minimum number of outgoing edges, $\delta(S)$, among every subset of vertices $S \subseteq V$. The celebrated result of Karger and Stein~\cite{karger1996new} solves \mc{} by recursively contracting edges in random order. Specifically, it runs two instances of the contraction process with different seeds in parallel. Each instance is run in parallel until the graph size is reduced by a factor of $\frac1{\sqrt{2}}$, at which point each instance recurses (thereby creating a parallel split again). They return the minimum of the two returned cuts. The algorithm itself is mainly inspired by another result of Karger~\cite{Karger93a} for finding the \mc{} using graph contractions.
We also extend our approach to the \mkc{} problem, in which we are given a graph $G = (V, E)$ and an integer $k$ and we want to find a decomposition of $V$ into $k$ subsets $V_1, V_2, \ldots, V_k$ so that $\sum_{i=1}^k{\delta(V_i)}$ is minimized. We utilize the greedy algorithm of Saran and Vazirani~\cite{SaranV95} which gives an $O(2-\frac2k)$-approximation of the \mkc{}. Gomuri and Hu give an alternative algorithm with the same approximation guarantee with additional features~\cite{GomoryHu61trees}. 

We study the \mc{} and \mkc{} problems in the $\AMPC$ model. We give  $O(\log\log n)$-round $\AMPC$ algorithms for a $(2+\epsilon)$-approximation of \mc{} and a $(4+\epsilon)$-approximation of \mkc{}.



\subsection{Adaptive Massively Parallel Computation ($\AMPC$)}

Massively Parallel Computation ($\MPC$) and Adaptive Massively Parallel Computation ($\AMPC$) both sprung out of an interest in formalizing a theoretical model for the famous MapReduce programming framework. The most common problems in MPC and AMPC are on graph inputs, and since our paper only considers graph problems, we define these two models in terms of problems on graphs. Consider a graph $G=(V,E)$ with $n=|V|$ and $m=|E|$.

In standard $\MPC$~\cite{goodrich2011sorting,andoni2014parallel,karloff2010a,lattanzi2011filtering}, we are given a collection of $\machines$ machines and are allowed to compute the solution to a problem in parallel. As we have already discussed, MPC computation occurs in synchronous rounds, each consisting of local polynomial-time computation and ending with machine-machine communication where all messages sent to and from a machine must fit within its local memory. Fully-scalable algorithms, the strongest memory regime in MPC, require the local memory to be constrained by $O(n^\epsilon)$ for any given $0<\epsilon<1$. Additionally, we are primarily interested in algorithms that require at most $O(\log n)$ rounds. However, often sublogarithmic -- i.e., $O(\sqrt{\log n})$ or $O(\log\log n)$ -- round complexity is much more desirable. In most cases, the total space must be at most $\widetilde{O}(n+m)$, though sometimes we allow slightly superlinear total space.

AMPC extends MPC to add functionality while remaining implementable on modern hardware. Formally, in the AMPC model, we are given a set of distributed hash tables $\mathcal{H}_0,\ldots,\mathcal{H}_k$ for each of the $k$ rounds of computation. These hash tables are each limited in size by the total space of the model (i.e., $\widetilde{O}(n+m)$). As in MPC, we are given a number of machines and computation proceeds in rounds. In each round, local computations occur and then messages are sent between machines. The distinction in AMPC is that during the local computations, machines are allowed simultaneous read access to the hash table for that round (i.e., $\mathcal{H}_{i-1}$ for round $i$) and during the messaging phase of the round, they are allowed to write data to the next hash table, $\mathcal{H}_i$. Reading and writing is limited by machine local memory. The power of the AMPC model over the MPC model is that, at the beginning of a round, the machines do not need to choose all the data they will access during the round. Instead, they can dynamically access the data stored in the hash table over the course of the local computation, thus potentially selecting data based on its own local computation. 

It is not too hard to see that AMPC is a strictly stronger model than MPC. In fact, it was formally shown that all MPC algorithms can be implemented in AMPC with the same round and space complexities~\cite{behnezhad2019massively}.


\subsection{Our Contributions and Methods}

This work is the first to study the Adaptive Massively Parallel Computation (AMPC) model for \mc{} problems on graphs. We mainly focus on the standard single \mc{} problem, although we also propose an approximation algorithm for the \mkc{} problem. Our main result for the \mc{} problem is a $2+\epsilon$ approximate algorithm that uses sublogarithmic $O(\log\log n)$ rounds.

\begin{theorem}\label{thm:mincut}
\thmmincut
\end{theorem}

Note that this is a vast improvement over the current state-of-the art algorithms in MPC by Ghaffari and Nowicki~\cite{ghaffari2020massively}, which achieves the same $2+\epsilon$ approximation in $O(\log n \log\log n)$ rounds. Both our algorithm and that of Ghaffari and Nowicki use Karger's methods as a general structure for finding the \mc{}.
Using this method, the goal is to recursively execute random graph contractions. From the results of Karger, the contraction process either finds a singleton cut that is a $2+\epsilon$ approximation or preserves a specific \mc{} with probability dependent on the depth of recursion. To leverage this result, at each step of the recursion process, we find the best singleton cut on the existing graph. Once the graph is small enough, the problem can be solved efficiently. Out of all the singleton solutions found during this process and the final \mc{} on the small graph, we simply select the best cut. This is a $2+\epsilon$ approximate \mc{} with high probability.

To implement this approach in a distributed model, both methods assign random weights to the edges of the input graph and find a minimum spanning tree (MST). Greedily, selecting edges in order of decreasing weight, we contract the graph along the current edge. This process is equivalent to the same greedy random contraction process on the original graph. This step, already, currently requires at least $\Omega(\log n)$ rounds in MPC, but the flexibility of the AMPC model allows us to achieve this step in a constant number of rounds.

It remains to show how can one find the best singleton cuts at each level of recursion. In order to do this, we employ a \emph{low-depth tree decomposition}
on the minimum spanning tree until it becomes a set of separated vertices. On top of this recursive divide-and-conquer process, we design a process to compute and remember the best singleton cut.

The high level idea of recursively partitioning the tree and applying a process on top of that to find the best singleton cut is the same in both our paper and Ghaffari and Nowicki's paper~\cite{ghaffari2020massively}. However, the processes used to do this in MPC do not yield simple improvements in AMPC. Rather, we must use entirely novel techniques that leverage adaptivity to get truly sublogarithmic results. In fact, this must be done in constant rounds to achieve our results, whereas Ghaffari and Nowicki do this in $O(\log n)$ rounds. In order to create a tree decomposition, we consider maximal paths of heavy edges (i.e., edges that go from a parent to its child with the largest subtree). These paths are replaced by binary trees whose leaves are the path and the root connects to the path's parent. Consider labeling the resulting vertices in the graph with their depth. For each internal node in one of these binary trees, which was \textit{not} a vertex in the original tree, we select a specific descendant leaf in the binary tree expansion of the path to send its depth to. The final value a vertex receives is then what we call the ``label'', which measures at what level of recursion the tree splits at that vertex. An entire labeling of the tree encodes an entire tree decomposition. This is done in constant AMPC rounds.

To compute the singleton cuts at each level, we assign to each singleton cut formed during the contraction process a vertex that has the lowest label. We show that such vertices are well-defined, i.e. there is only one vertex with the lowest label within vertices on the same side of a singleton cut. Because removing vertices of labels lower than $i$ partitions the tree into disjoint subtrees such that each subtree contains at most one vertex with label $i$, we are able to calculate minimal singleton cuts corresponding to these vertices with label $i$ in parallel in a constant number of AMPC rounds. Since, we constructed the low-depth decomposition such that the range of labels has size $O(\log^{2}n)$, thus, by increasing the total memory, we can perform these computations for all different lables in a constant number of AMPC rounds. For more details, we defer to Section~\ref{sec:singleton}.


We then show how this work can be leveraged to achieve efficient results for approximate \mkc{}, generalizing the results from Saran and Vazirani~\cite{SaranV95}. At a high level, we start by computing a \mc{}. Then we add the edges of the cut to a set $D$. In all following $k-1$ iterations, we calculate the \mc{} on the graph without edges in $D$, and add the new cut edges to $D$ for the next iteration. The set of the first $k$ cuts we compute is our $k$-cut. 

Compared to Saran's and Vazirani's technique, our method uses an \emph{approximate} \mc{} rather than an exact \mc{} on each splitting step. This requires adapted analysis of this general approach. We employ the structure of Gomory-Hu trees (see ~\cite{GomoryHu61trees}) for this purpose and show the following result:

\begin{theorem}\label{thm:vazirani'scuts}
\thmvazirani
\end{theorem}

Therefore, for small values of $k$, we can achieve efficent algorithms for $4+\epsilon$ approximate \mkc{} in AMPC. Note that there are no existing results in the MPC model, however our methods applied to the work of Ghaffari and Nowicki~\cite{ghaffari2020massively} yield:

\begin{corollary}
There is an algorithm that achieves a $(4 + \epsilon)$ approximation of the \mkc{} with high probability that can be implemented in the MPC model with $\bigO(n^{\epsilon})$ memory per machine in $\bigO(k \log n\log\log{n})$ rounds and $\bigO(m)$ total memory.
\end{corollary}

Note there is still a logarithmic-in-$n$ improvement in the round complexity in AMPC over MPC no matter the value of $k$. Due to space constraints both these result are presented in the appendix.

\section{Minimum Cut in AMPC}

\newcommand{\levels}{\mathcal{L}}
\newcommand{\opt}{\mathsf{OPT}}



Karger and Stein~\cite{karger1996new} proposed a foundational edge contraction strategy for solving \mc{}: 
\begin{itemize}
    \item Create two copies of $G$, and independently on each, contract edges in a random order until there are at most $\frac{n}{\sqrt{2}}$ vertices.
    \item Recursively solve the problem on each contracted copy until they have constant size.
    \item Return the minimum of the cuts found on both copies.
\end{itemize}


\begin{lemma}[\cite{karger1996new}]\label{lem:karger}
The contraction process executed to the point where there are only $\frac{n}{t}$
vertices left preserves any fixed minimum cut with probability $\Omega\left(\frac{1}{t^2}\right)$.
\end{lemma}

According to Lemma~\ref{lem:karger}, naively contracting random edges until there are only two vertices remaining preserves at least one minimum cut with probability $\Omega\left(\frac{1}{n^2}\right)$. Thus, we need to repeat the naive contraction process at least $O\left(n^2\log n\right)$ times so that we have a high \emph{probability of success}, i.e., preserving a minimum cut. However, Karger and Stein~\cite{karger1996new} show that their recursive strategy  succeeds with probability $\Omega\left(\frac{1}{\log n}\right)$. In turn, running $O\left(\log^2 n\right)$ instances of the recursive strategy is enough to find a minimum cut with high probability. 

Roughly speaking, the choice of $t = \sqrt{2}$ as the inverse of the branching factor assures that a minimum cut is preserved with probability $\frac{1}{t^2} = \frac{1}{2}$ throughout the contractions in each copy. Thus, the probability of success, say $P(n)$, for $n$ vertices is bounded by:
\begin{align} \label{eq:probsucc}
P(n) \geq 1 - \left(1 - \frac{1}{2}\cdot P\left(\frac{n}{\sqrt{2}}\right)\right)^2
\end{align}

Note that the random contractions in two copies are assumed to be independent, and the probability of success for each copy is at least $\frac{1}{2}\cdot P\left(\frac{n}{\sqrt{2}}\right)$ since we recurse on the resulting contracted graph with $\frac{n}{\sqrt{2}}$ vertices.
 Inequality (\ref{eq:probsucc}) implies that at the $k$-th level of recursion (counting from the bottom), the probability of success is $\Omega\left(\frac{1}{k}\right)$, and in particular $\Omega\left(\frac{1}{\log n}\right)$ at the root of recursion.~\cite{karger1996new}. 

Let us now give some high-level insight into the approach by Ghaffari's and Nowicki. Ghaffari and Nowicki~\cite{ghaffari2020massively} observed that if we only desire a $(2 + \epsilon)$ approximate cut, we can use a better bound for the probability of preserving a minimum cut, or alternatively, the success probability.

\begin{lemma}[\cite{ghaffari2020massively,karger1996new}]\label{lemma:success}
On an $n$-vertex graph $G$, let $C$ be a minimum cut with weight $\lambda$. Fix an arbitrary $\epsilon \in (0,1)$. The described random
contraction process that contracts $G$ down to $\frac{n}{t}$
vertices either at some step creates a singleton cut of size at most
$(2 + \epsilon)\lambda$ or preserves $C$ - i.e., it does not contract any of its edges - with probability at least $\frac{1}{t^{1 - \epsilon / 3}}$.
\end{lemma}

A \emph{singleton cut} is a partitioning of graph vertices so that there is only one vertex on one side, i.e., $\delta(S)$ so that $|S| = 1$. Assuming that one is able to verify whether a singleton cut of a small size has been formed during the contraction process, they show that this greater probability of success can boost the recursive process. In short, consider the $k$-th level of recursion, where level $0$ corresponds to the bottom level. Let $\frac{n}{t_{k}}$ be the size of a single recursive instance  at level $k$, and denote by $s_{k}$ the total number of instances on this level. For all $k$, they ensure $s_{k} = t_{k}^{1 - \epsilon / 3} $. 

Now, let $x_{k}^{1 - \epsilon / 3}$ be the branching factor on level $k$. That is, the recursion produces $x_{k}^{1 - \epsilon / 3}$ copies of the instance at level $k$, and on each of them independently contracts edges in a random order until the number of vertices is bigger than $\frac{n}{t_{k}} \cdot \frac{1}{x_{k}}$. If we have an algorithm that is able to \textit{track} whether a small singleton cut appeared in each of these random processes, we either get a singleton cut that $(2+\epsilon)$ approximates a minimum cut or a minimum cut is preserved with probability $x_{k}^{1 - \epsilon / 3}$. Since we made $x_{k}^{1-\epsilon / 3}$ copies, by a similar argument as in Karger's approach, we get that, in the latter case, the probability of preserving a minimum cut is $\Omega\left(\frac{1}{k}\right)$. 

Finally, observe that on the $k$-th level of recursion, the most costly operation is copying a $k$-th level instance $x_{k}^{1 - \epsilon / 3}$ times in order to contract edges in each of these instances. Since the instance has size $\frac{n}{t_{k}}$ and we have $s_{k}$ instances, processing these tasks in parallel requires $\frac{n}{t_{k}} \cdot s_{k} \cdot x_{k}^{1-\epsilon / 3}$ space. If one want to fit this in $O(n)$ space, then it must be that $x_{k} \le t_{k}^{(\epsilon / 3) / (1 - \epsilon / 3)}$. Anyway, we get that the number of contractions we can make on $k$-th level is polynomial in the number of contractions we made on higher levels, and if the recurrence is solved, then it follows that it will be $O(\log\log n)$ levels until we reach a graph of a constant size.

Ghaffari and Nowicki~\cite{ghaffari2020massively}, use Lemma~\ref{lemma:success} and the above boosting scheme to show an $O(\log\log{n}\cdot\log{n})$-round MPC algorithm for \mc{}. The main non-trivial part of their algorithm involves tracking the smallest singleton cut on each recursion level, which they do in $O(\log{n})$ rounds because of the divide and conquer nature of their approach. Effectively, they assign all edges random and unique edge weights, and contract all uncontracted edges in decreasing order by edge weight. It can then be shown that all that needs to be done is to compute the MST of this graph and contract these edges accordingly (all other edges will be automatically contracted when another edge is contracted). We reduce the number of rounds for singleton cut tracking down to $O(1)$ rounds in the AMPC model. We aim to prove the following theorem.

\begin{customthm}{\ref{thm:mincut}}
\thmmincut
\end{customthm}

To track singleton cuts, the first step is to find a \emph{low depth decomposition} of the current MST. At a high level, a low depth decomposition of a tree is a labeling of its vertices with values $1$ through $d$, where $d$ is the depth. This label must satisfy the following: for every level $i\in[d]$, the connected components induced on vertices with label at least $i$ must contain at most one vertex for each $i$. This defines a recursive splitting process: starting at depth 1, there must be at most one vertex $v$ with the minimum label, so we can split the tree into multiple parts by removing $v$. Then we simply recurse on each connected component, considering the next set of labels, and knowing the process will always split each connected component once at a time. This is the general idea captured by both this and previous works. However, in order to increase the efficiency of this step, we require a new decomposition structure (see Definition~\ref{def:lowdepth}) and new methods for finding the decomposition. Notice that it is always true that at each level, each connected component contains at most one vertex at the next level.

In Section~\ref{sec:low-depth}, we show how to find a low depth decomposition with depth $O(\log^2{n})$ in AMPC in $O(1/\epsilon)$ rounds (Lemma~\ref{lem:treedecomp}) with $O(n^\epsilon)$ space per machine. Roughly speaking, we create a heavy-light decomposition of the MST, where we store ``heavy paths'' consisting of edges connecting vertices to their children with the largest number of descendants and isolated ``light nodes''. We replace each heavy path with a complete binary tree whose leaves contain the vertices in the path, which gives us an efficient structure to obtain our labeling. This yields our low depth decomposition.

In the next step, we compute the size of
of every singleton cut $S$ that is created during the process. Note that the contractions are inherently sequential and the number of contractions we need to make at step $k$ is $x_{k} \in O(n)$. However, each singleton cut is a connected component on the MST containing a specific edge $e$, whose contraction -- in the increasing order of contracting MST edges -- results in subset $S$, if we only allow the edges that have a smaller weight than $e$. We partition these connected components based on the vertex in the cut with the lowest level in the heavy-light decomposition of the MST. This way, we can compute every level of the low depth decomposition in parallel with only an $O(\log^2{n})$ blowup in total memory. In Section~\ref{sec:singleton}, we show that we can track every singleton cut in the contraction process in $O(1)$ AMPC rounds. A high level pseudocode of the main algorithm is given in Algorithm~\ref{alg:min_cut_full}.

\begin{algorithm}[ht]
\SetAlgoLined
\KwData{A graph $G = (V(G), E(G))$, a parameter $k$.}
\KwResult{$(2 + \epsilon)$ approximation of \mc{}.}
\If{$|G| \in n^{\epsilon}$}
{\KwRet{\mc{} of $G$ calculated on a single machine}}
Let $\widehat{G_{1}}$, $\ldots$, $\widehat{G_{k}}$ be copies of $G$ with assigned random weight on edges (independently for each copy)\;
\underline{In parallel for all $i\in[k]$, $S_{i}$ $\leftarrow$ \textsf{MinSingletonCut}$(\widehat{G_{i}})$}\;
In parallel for all $i\in[k]$, $G_{i}$ $\leftarrow$ copy of $\widehat{G_{i}}$ after first $k$ contractions\;
In parallel, $C_{i}$ $\leftarrow$ \textsf{AMPC-MinCut}$(G_{i})$\;
\KwRet{$\min(S_{1}, \ldots, S_{k}, C_{1}, \ldots, C_{k})$}\;
\caption{\textsf{AMPC-MinCut} \\(An algorithm that calculates $(2 + \epsilon)$ approximation of \mc{} in $G$. The novel part is underlined. )}
\label{alg:min_cut_full}
\end{algorithm}

Note that \textsf{MinSingletonCut} (Algorithm~\ref{alg:smallestsingletoncut}) is introduced in Section~\ref{sec:singleton} and it leverages \textsf{LowDepthDecomp} (Algorithm~\ref{alg:lowdepthdecomp}) from Section~\ref{sec:low-depth}.

\section{Generalized Low Depth Tree Decomposition}\label{sec:low-depth}

This section and the next address our algorithmic formulation and analysis. Note that all omitted proofs are deferred to the Appendix.

In order to efficiently compute the singleton cuts in parallel, we first need to compute an efficient decomposition of the MST. The low depth tree decomposition Ghaffari and Nowicki~\cite{ghaffari2020massively} introduce is a very specific decomposition with $i$ levels such that at each level $\ell$, any connected component of size $s$ on vertices at that level or higher has a single vertex at level $\ell$ that separates the component into two components with size at least $s/3$ each. Unfortunately, it is unclear how to calculate this precise decomposition efficiently in AMPC. To work around this, we introduce a more generalized version of the low depth tree decomposition, show that it can be computed in AMPC, and later show that we can leverage this to obtain our \mc{} algorithm.

\begin{definition}\label{def:lowdepth}
A \textbf{generalized low depth tree decomposition} of some tree $T$ is a labeling $\ell:V(T) \to [h]$ of vertices with \textbf{levels} for decomposition height $h\in O(\log^2n)$ such that for each level $i$, the connected components induced on $T^i = \{v\in T: \ell(v) \geq i\}$ have at most one vertex labeled $i$ each.
\end{definition}

Notice we do not define how a level is assigned; we simply require it is assigned to satisfy the property on connected components. We describe one way to do that in this section.

To see what such a decomposition looks like, consider a process where at timestep $t$ we look at the subgraph induced on the vertices $v$ with $\ell(v) \geq t$ (i.e., $T^t$). Consider a connected component $C$ and let $v$ be its minimum level vertex. Then $\ell(v) \geq t$, and it is the only vertex at that level in $C$. At timestep $\ell(v) + 1$, $C$ becomes separated into multiple components who all contain a vertex adjacent to $v$. This process defines forests with smaller and smaller trees as time passes, and eventually results in isolated vertices. The completion time of this process depends on the height of the decomposition, which in our case is $O(\log^2n)$. We will, of course, make this more efficient in Section~\ref{sec:singleton}.

It is not that hard to see that Ghaffari and Nowicki's low depth tree decomposition is a specific example of generalized tree decomposition with depth $O(\log n)$. They put a single vertex in the first level and then simply recurse on the two trees in the remaining forest. Note that they require additional properties of this decomposition to obtain their result, specifically that each new component has size at least $\frac13$ of the original component, but we will see later that these are not necessary for finding the singleton cuts.

Like Ghaffari and Nowicki in MPC, we prove this can be computed efficiently in, instead, AMPC. 

\begin{lemma}\label{lem:treedecomp}
Computing a generalized low depth tree decomposition of an $n$-vertex tree can be done in $O(1/\epsilon)$ AMPC rounds with $O(n^\epsilon)$ memory per machine and $O(n\log^2 n)$ total memory.
\end{lemma}

The rest of this section is dedicated to proving Lemma~\ref{lem:treedecomp}. The formal and complete algorithm is shown in Algorithm~\ref{alg:lowdepthdecomp} and further details and definitions can be found later in this section. At a high level, our algorithm proceeds as follows:
\begin{enumerate}
\item Root the tree and orient the edges [line~\ref{lin:root}].
\item Contract heavy paths in a \textit{heavy-light decomposition} of $T$ into \textit{meta vertices} to construct a \textit{meta tree}, $T_M$ [lines~\ref{lin:heavylightstart} to~\ref{lin:heavylightend}].
\item For each meta vertex, create a \textit{binarized path}, a binary tree whose leaves are the vertices in the heavy path, in order. Expanding meta vertices in this manner yields our \textit{expanded meta tree} [lines~\ref{lin:expandvertsstart} to~\ref{lin:expandvertsend}].
\item Label each vertex according to properties of the expanded meta tree. For all new vertices (i.e., vertices created in step 3) $v$, label $v$ with the depth of the highest vertex $u$ in the same meta node such that $v$ is the leftmost leaf descending from the rightmost child of $u$ in the binarized path of the meta node [lines~\ref{lin:labelvertsstart} to~\ref{lin:labelvertsend}].
\end{enumerate}

Each of these steps correspond to the following subsections. For instance, step 1 corresponds to Section~\ref{sec:low-depth}.1. All relevant terminology related to these steps are additionally found in the corresponding subsections. Lemma~\ref{lem:treedecomp} is proven at the end of the final subsection.

\begin{algorithm}[ht]
\SetAlgoLined
\KwData{A tree $T = (V(T), E(T))$.}
\KwResult{A mapping $\ell:V(T)\to \mathbb{N}$ of tree vertices to levels.}
 Initialize $\ell:V(T)\to \mathbb{N}$\;
 Root and orient $T$\;\label{lin:root}
 Let $T_H = (V(T), \{e\in E(T): e\text{ is heavy}\}$\;\label{lin:heavylightstart}
 Let $\mathcal{P}$ be the connected components of $T_H$\;
 Let $T_M = (\mathcal{P}, \{(P_1,P_2): P_1,P_2\in\mathcal{P}, \exists (u_1,u_2)\in V(P_1)\times V(P_2) \text{ such that } (u_1,u_2) \in E(T)\})$\;\label{lin:heavylightend}
 \For{$v\in T_M$ of heavy path $P_v$ in parallel} {
    Let $V(T_v)$ be a vertex set of size $2|P_v|-1$ with associated indices $1,\ldots,2|P_v|-1$, denoted by $i_u$\;\label{lin:bintreestart}\label{lin:expandvertsstart}
    Let $T_v = (V(T_v), \{(u,p_u): i_{p_u} = \lfloor i_u /2\rfloor \})$\;\label{lin:bintreeend}
    Pre-order traverse $T$ and sort $P_v$ accordingly\;\label{lin:mapleavesstart}
    Pre-order traverse $T_v$ and let $L$ be its sorted leaves\;
    For all $i\in [|P_v|]$, map $P_v[i]$ to $L[i]$\;\label{lin:mapleavesend}\label{lin:expandvertsend}
    \For{$u \in V(T_v)$}{
        Find path $P^u$ to the root of the expanded meta-tree\;\label{lin:labelvertsstart}\label{lin:findpath}
        Let $u'\in V(T_v)\cap P^u$ be such that $u$ is the leftmost descendant of $u'$'s right child (otherwise $u'=u$)\;\label{lin:findleftright}
        Label $\ell(u) = d(u')$\;\label{lin:labelvertsend}\label{lin:label}
    }
  }
  \KwRet{$\ell$ limited to the original vertices in $T$}\;
 \caption{\textsf{LowDepthDecomp} \\(Computing a generalized low depth tree decomposition of an input tree in AMPC)}
\label{alg:lowdepthdecomp}
\end{algorithm}

\subsection{Rooting the Tree}
Like in Ghaffari and Nowicki, the first thing we need to do in line~\ref{lin:root} of Algorithm~\ref{alg:lowdepthdecomp} is compute an orientation of the edges. Fortunately this, along with rooting the tree, can be done quickly in AMPC by the results of Behnezhad et al.~\cite{behnezhad2019massively} in their Theorem 7.

\begin{lemma}[Behnezhad et al.~\cite{behnezhad2019massively}]\label{lem:orient}
Given a forest $F$ on $n$ vertices, the trees in $F$ can be rooted and edges can be oriented in $O(1/\epsilon)$ AMPC rounds w.h.p. using $O(n^\epsilon)$ local memory and $O(n\log n)$ total space w.h.p.
\end{lemma}

Here, w.h.p. means ``with high probability.'' This completes the first step of our algorithm.

\subsection{Meta Tree Construction}\label{sec:metatree}

We also leverage Ghaffari and Nowicki's notion of heavy-light decompositions for our AMPC algorithm, which can be found from lines~\ref{lin:heavylightstart} through~\ref{lin:heavylightend} in Algorithm~\ref{alg:lowdepthdecomp}. This process allows us to quickly decompose the tree into a set of disjoint paths of \textit{heavy edges}, which are defined as follows (note that our definition slightly deviates from Ghaffari and Nowicki~\cite{ghaffari2020massively}, where the heavy edge must extend to the child with the largest subtree without requiring this subtree to be that large, though it is the same as the definition used by Sleator and Tarjan~\cite{sleator1981a}):

\begin{definition}[Sleator and Tarjan~\cite{sleator1981a}]
Given a tree $T$ and a vertex $v\in T$, let $\{u_i\}_{i\in k}$ be the set of children of $v$ where the subtree rooted at $u_1$ is the largest out of all $u_i$. If there is no strictly largest subtree, we arbitrarily choose exactly one of the children with a largest subtree. Then $(u_1,v)$ is a \textbf{heavy edge} and $(u_i,v)$ is a \textbf{light edge} for all $1 < i \leq k$.
\end{definition} 

\begin{figure}
    \centering
    \scalebox{0.57}{\input{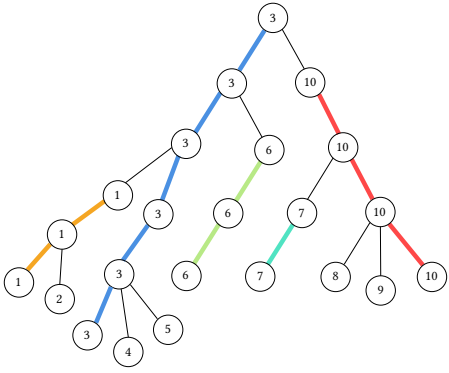}}
    \caption{The heavy-light decomposition of an example tree.}
    \label{fig:heavylight}
\end{figure}

Then the definition of a heavy path follows quite simply.

\begin{definition}[Ghaffari and Nowicki~\cite{ghaffari2020massively}]
Given a tree $T$, a \textbf{heavy path} is a maximal length path consisting only of heavy edges in $T$.
\end{definition}

Ghaffari and Nowicki then make the observation that the number of light edges and heavy paths is highly limited in a tree. This comes from a simple counting argument, where if you consider the path from root $r$ to some vertex $v$, any time you cross a light edge, the size of the current subtree is reduced by at least a factor of 2. This holds even with our different notion of heavy edges since subtrees rooted at children of light edges are still much smaller compared to the subtree rooted at the parent vertex. This bounds the number of light edges between $r$ and $v$, where each pair of light edges are separated by at most one heavy path, and therefore it also bounds the number of heavy paths.

\begin{observation}[Ghaffari and Nowicki~\cite{ghaffari2020massively}]\label{obs:lightheavy}
Consider a tree $T$ oriented towards root $r$. For each vertex $v$, there are only $O(\log n)$ light edges and only $O(\log n)$ heavy paths on the path from $v$ to $r$.
\end{observation}

Using the definition of heavy edge from Sleator and Tarjan~\cite{sleator1981a} instead of from Ghaffari and Nowicki, we get an additional nice property. This is because in our definition, every internal vertex has one descending heavy edge to one child.

\begin{observation}[Sleator and Tarjan~\cite{sleator1981a}]\label{obs:heavypath}
Given a tree $T$ and an internal vertex $v\in T$, $v$ must be on exactly one heavy path. For a leaf $\ell \in T$, $\ell$ must be on at most one heavy path.
\end{observation}

Our first goal is to compute what we call a \textit{meta tree}. This is a decomposition of our tree that will allow us to  effectively handle heavy edges. It is quite analogous to Ghaffari and Nowicki's notion of the heavy-light decomposition, which partitions the tree into heavy and light edges.

\begin{definition}
Given a tree $T$, the \textbf{meta tree} of $T$, denoted $T_M$, comes from contracting all the heavy paths in $T$. We call the vertices of $T$ \textbf{original vertices} and the vertices of $T_M$ \textbf{meta vertices}.
\end{definition}

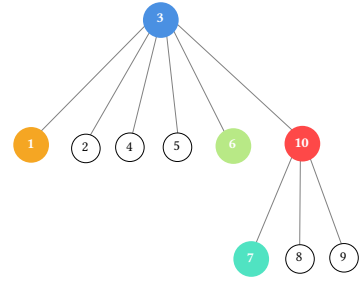
\begin{figure}
    \centering
    \scalebox{0.6}{\tikzset{every picture/.style={line width=0.45pt}} 

\begin{tikzpicture}[x=0.75pt,y=0.75pt,yscale=-1,xscale=1]

\draw [color={rgb, 255:red, 128; green, 128; blue, 128 }  ,draw opacity=1 ]   (387,185.51) -- (350,274.18) ;
\draw [color={rgb, 255:red, 128; green, 128; blue, 128 }  ,draw opacity=1 ][fill={rgb, 255:red, 255; green, 255; blue, 255 }  ,fill opacity=1 ]   (391.33,187.84) -- (390.67,275.18) ;
\draw [color={rgb, 255:red, 128; green, 128; blue, 128 }  ,draw opacity=1 ][fill={rgb, 255:red, 255; green, 255; blue, 255 }  ,fill opacity=1 ]   (397.67,188.18) -- (429,274.18) ;
\draw [color={rgb, 255:red, 128; green, 128; blue, 128 }  ,draw opacity=1 ]   (261.67,80.84) -- (171.33,172.51) ;
\draw [color={rgb, 255:red, 128; green, 128; blue, 128 }  ,draw opacity=1 ]   (267.67,81.84) -- (211.67,179.18) ;
\draw [color={rgb, 255:red, 128; green, 128; blue, 128 }  ,draw opacity=1 ]   (272,84.18) -- (248.33,180.18) ;
\draw [color={rgb, 255:red, 128; green, 128; blue, 128 }  ,draw opacity=1 ]   (278.33,84.51) -- (289,179.51) ;
\draw [color={rgb, 255:red, 128; green, 128; blue, 128 }  ,draw opacity=1 ]   (283.33,83.18) -- (330,176.18) ;
\draw [color={rgb, 255:red, 128; green, 128; blue, 128 }  ,draw opacity=1 ]   (286,78.84) -- (388.67,173.51) ;
\draw  [draw opacity=0][fill={rgb, 255:red, 74; green, 144; blue, 226 }  ,fill opacity=1 ] (259,77) .. controls (259,68.72) and (265.72,62) .. (274,62) .. controls (282.28,62) and (289,68.72) .. (289,77) .. controls (289,85.28) and (282.28,92) .. (274,92) .. controls (265.72,92) and (259,85.28) .. (259,77) -- cycle ;\draw (237,70) node [anchor=north west,text=white][inner sep=0.75pt]   [align=left] {\begin{minipage}[lt]{52.6pt}\setlength\topsep{0pt}
\begin{center}
\textbf{3}
\end{center}

\end{minipage}};
\draw  [draw opacity=0][fill={rgb, 255:red, 245; green, 166; blue, 35 }  ,fill opacity=1 ] (150,182) .. controls (150,173.72) and (156.72,167) .. (165,167) .. controls (173.28,167) and (180,173.72) .. (180,182) .. controls (180,190.28) and (173.28,197) .. (165,197) .. controls (156.72,197) and (150,190.28) .. (150,182) -- cycle ;\draw (128,176) node [anchor=north west, text=white][inner sep=0.75pt]   [align=left] {\begin{minipage}[lt]{52.6pt}\setlength\topsep{0pt}
\begin{center}
\textbf{1}
\end{center}

\end{minipage}};
\draw  [color={rgb, 255:red, 0; green, 0; blue, 0 }  ,draw opacity=1 ][fill={rgb, 255:red, 255; green, 255; blue, 255 }  ,fill opacity=1 ] (199,185) .. controls (199,178.37) and (204.37,173) .. (211,173) .. controls (217.63,173) and (223,178.37) .. (223,185) .. controls (223,191.63) and (217.63,197) .. (211,197) .. controls (204.37,197) and (199,191.63) .. (199,185) -- cycle ;\draw (174,178) node [anchor=north west][inner sep=0.75pt]   [align=left] {\begin{minipage}[lt]{52.6pt}\setlength\topsep{0pt}
\begin{center}
$2$
\end{center}

\end{minipage}};
\draw  [draw opacity=0][fill={rgb, 255:red, 254; green, 72; blue, 72 }  ,fill opacity=1 ] (378,181) .. controls (378,172.72) and (384.72,166) .. (393,166) .. controls (401.28,166) and (408,172.72) .. (408,181) .. controls (408,189.28) and (401.28,196) .. (393,196) .. controls (384.72,196) and (378,189.28) .. (378,181) -- cycle ;\draw (356,175) node [anchor=north west,text=white][inner sep=0.75pt]   [align=left] {\begin{minipage}[lt]{52.6pt}\setlength\topsep{0pt}
\begin{center}
\textbf{10}
\end{center}

\end{minipage}};
\draw  [draw opacity=0][fill={rgb, 255:red, 80; green, 227; blue, 194 }  ,fill opacity=1 ] (335,278) .. controls (335,269.72) and (341.72,263) .. (350,263) .. controls (358.28,263) and (365,269.72) .. (365,278) .. controls (365,286.28) and (358.28,293) .. (350,293) .. controls (341.72,293) and (335,286.28) .. (335,278) -- cycle ;\draw (313,272) node [anchor=north west,text=white][inner sep=0.75pt]   [align=left] {\begin{minipage}[lt]{52.6pt}\setlength\topsep{0pt}
\begin{center}
\textbf{7}
\end{center}

\end{minipage}};
\draw  [draw opacity=0][fill={rgb, 255:red, 184; green, 233; blue, 134 }  ,fill opacity=1 ] (320,183) .. controls (320,174.72) and (326.72,168) .. (335,168) .. controls (343.28,168) and (350,174.72) .. (350,183) .. controls (350,191.28) and (343.28,198) .. (335,198) .. controls (326.72,198) and (320,191.28) .. (320,183) -- cycle ;\draw (298,176) node [anchor=north west,text=white][inner sep=0.75pt]   [align=left] {\begin{minipage}[lt]{52.6pt}\setlength\topsep{0pt}
\begin{center}
\textbf{6}
\end{center}

\end{minipage}};
\draw  [color={rgb, 255:red, 0; green, 0; blue, 0 }  ,draw opacity=1 ][fill={rgb, 255:red, 255; green, 255; blue, 255 }  ,fill opacity=1 ] (236,184) .. controls (236,177.37) and (241.37,172) .. (248,172) .. controls (254.63,172) and (260,177.37) .. (260,184) .. controls (260,190.63) and (254.63,196) .. (248,196) .. controls (241.37,196) and (236,190.63) .. (236,184) -- cycle ;\draw (211,178) node [anchor=north west][inner sep=0.75pt]   [align=left] {\begin{minipage}[lt]{52.6pt}\setlength\topsep{0pt}
\begin{center}
$4$
\end{center}

\end{minipage}};
\draw  [color={rgb, 255:red, 0; green, 0; blue, 0 }  ,draw opacity=1 ][fill={rgb, 255:red, 255; green, 255; blue, 255 }  ,fill opacity=1 ] (276,184) .. controls (276,177.37) and (281.37,172) .. (288,172) .. controls (294.63,172) and (300,177.37) .. (300,184) .. controls (300,190.63) and (294.63,196) .. (288,196) .. controls (281.37,196) and (276,190.63) .. (276,184) -- cycle ;\draw (251,178) node [anchor=north west][inner sep=0.75pt]   [align=left] {\begin{minipage}[lt]{52.6pt}\setlength\topsep{0pt}
\begin{center}
$5$
\end{center}

\end{minipage}};
\draw  [color={rgb, 255:red, 0; green, 0; blue, 0 }  ,draw opacity=1 ][fill={rgb, 255:red, 255; green, 255; blue, 255 }  ,fill opacity=1 ] (379,278) .. controls (379,271.37) and (384.37,266) .. (391,266) .. controls (397.63,266) and (403,271.37) .. (403,278) .. controls (403,284.63) and (397.63,290) .. (391,290) .. controls (384.37,290) and (379,284.63) .. (379,278) -- cycle ;\draw (354,271) node [anchor=north west][inner sep=0.75pt]   [align=left] {\begin{minipage}[lt]{52.6pt}\setlength\topsep{0pt}
\begin{center}
$8$
\end{center}

\end{minipage}};
\draw  [color={rgb, 255:red, 0; green, 0; blue, 0 }  ,draw opacity=1 ][fill={rgb, 255:red, 255; green, 255; blue, 255 }  ,fill opacity=1 ] (416,277) .. controls (416,270.37) and (421.37,265) .. (428,265) .. controls (434.63,265) and (440,270.37) .. (440,277) .. controls (440,283.63) and (434.63,289) .. (428,289) .. controls (421.37,289) and (416,283.63) .. (416,277) -- cycle ;\draw (391,271) node [anchor=north west][inner sep=0.75pt]   [align=left] {\begin{minipage}[lt]{52.6pt}\setlength\topsep{0pt}
\begin{center}
$9$
\end{center}

\end{minipage}};

\end{tikzpicture}}
    \caption{The \emph{meta-tree} of the same tree from Figure~\ref{fig:heavylight} is demonstrated in this figure.}
    \label{fig:metatree}
\end{figure}

Note that contracting all heavy paths simultaneously is valid because, by Observation~\ref{obs:heavypath}, all heavy paths must be disjoint. Additionally, all internal meta vertices are contracted heavy paths (as opposed to original vertices), again by Observation~\ref{obs:heavypath}. We note that in AMPC, since connectivity is easy, it is additionally quite easy to contract the heavy paths of $T$ into single vertices.

\begin{lemma}\label{lem:heavylight}
Given a tree $T$, the meta tree $T_M$ can be computed, rooted, and oriented in AMPC in $O(1/\epsilon)$ rounds with $O(n^\epsilon)$ memory per machine and $O(n\log^2n)$ total space w.h.p.
\end{lemma}

This completes the second step of our decomposition algorithm.

\subsection{Expanding Meta Vertices}

In order to label the vertices, we need a way to handle the heavy paths corresponding to each meta vertex. Let $v\in T_M$ be a meta vertex, and $P_v$ be the heavy path of original vertices in $T$ corresponding to $v$. Note that we have no stronger bound on the length of a heavy path than $O(n)$. Therefore, a recursive partitioning, or labeling of vertices that has polylogarithmic depth must be able to cleverly divide heavy paths. We can do this with a new data structure.

\begin{definition}
Given some path $P$, a \textbf{binarized path} is an almost complete binary tree $T$ with $|P|$ leaves where there is a one-to-one mapping between $P$ and the leaves of $T$ such that the pre-order traversal of $P$ and $T$ limited to its leaves agree.
\end{definition}

By ``agree'', we mean that if a vertex $v\in P$ comes before a vertex $u\in P$ in the pre-order traversal of $P$, then it also does in the pre-order traversal of $T$. To characterize this tree, we make a quick observation: 

\begin{observation}\label{obs:binarytree}
An almost complete binary tree on $n$ leaves has $2n-1$ vertices, $\lfloor \log_2n\rfloor + 1$ max depth, and every layer is full except the last, which has $2n-2^{\lfloor\log_2 n\rfloor+1}$ vertices.
\end{observation}

Additionally, we can find a relationship between the ancestry of triplets in $P$ based off of the order of the three vertices. While this is not required for expanding meta vertices, it is a property of the binarized path that will be useful when we label vertices later.

\begin{observation}\label{obs:ancestors}
Given a binarized path $T$ of a path $P$, for any $u,u',u''\in P$ that appear in that order (or reversed), if $v$ is the lowest common ancestor of $u$ and $u'$ and $v'$ is the lowest common ancestor of $u$ and $u''$, then $v'$ is an ancestor of $v$ or $v'=v$.
\end{observation}

To create the tree, we do the following for every $v\in T_M$:
\begin{enumerate}
\item Create an almost complete binary tree $T_v$ with $|P_v|$ leaves, linking children to parents and noting if a vertex is a left or right child [lines~\ref{lin:bintreestart} and~\ref{lin:bintreeend}].
\item Do a pre-order traversal of $T_v$ and $P_v$ and map the vertices in $P_v$ to the leaves of $T_v$ such that the pre-order traversal of $P_v$ and of $T_v$ limited to its leaves agree. [lines~\ref{lin:mapleavesstart} to~\ref{lin:mapleavesend}].
\end{enumerate}

Next, it is pretty direct to see that the produced tree is a binarized path.

\begin{observation}\label{obs:binpathcorrect}
The process described above produces a binarized path $T_v$ of $P_v$ for all $v$.
\end{observation}

We prove that this can be done in the proper constraints.

\begin{lemma}\label{lem:expandverts}
The heavy paths of a tree can be converted into binarized paths in $O(1/\epsilon)$ AMPC rounds with $O(n^\epsilon)$ local memory and $O(n\log n)$ total space w.h.p.
\end{lemma}

\subsection{Labeling Vertices}

Our next goal is to label the vertices with the level they should be split on. Consider, hypothetically, expanding the meta tree $T_M$ such that every heavy path for a meta vertex $v$ is replaced with its binarized path (which is an almost complete binary tree) $T_v$, and the tree continues at the leaves corresponding to the nodes in the heavy path. Note that only some vertices in the hypothetical tree correspond to vertices in the original tree $T$. Specifically, the internal nodes of each component subtree $T_v$ are not vertices in $T$, but the leaves correspond exactly to the vertices in $T$.

Ultimately, for a vertex $u\in T$ in meta vertex $v$, let $u'$ be the vertex in $T_{v}$ such that $u$ is the leftmost leaf-descendant of the right child of $u'$ in $T_{u_M}$ (or if this doesn't exist, $u'=u$). Then we will label $\ell(u) = d(u')$ where $d$ is the depth in the expanded meta tree. Following this, our vertex labeling process will be as follows for each $v\in T_M$ and $u\in T_v$:
\begin{enumerate}
\item Vertex $u$ finds the path $P^u$ from $u$ to the root of $T_M$, assuming the meta vertices are expanded [line~\ref{lin:findpath}].
\item Let $u'$ be the highest vertex in $T_v$ such that $u$ is the leftmost descendant of the right child of $u$. If there is no such vertex, let $u'=u$ [line~\ref{lin:findleftright}].
\item Label $u$ with the depth (assuming roots have depth 1) of $u'$ in the expanded $T_M$ [line~\ref{lin:label}].
\end{enumerate}

We start by making a quick observation that comes directly from Observations~\ref{obs:lightheavy} and~\ref{obs:binarytree}.

\begin{observation}\label{obs:depth}
The max depth of $T_M$ with meta nodes expanded (``the expanded $T_M$'') into binary trees is $O(\log^2n)$.
\end{observation}

This will be greatly helpful in showing the efficiency of our algorithm. We now show that this final part can be implemented efficiently, which is sufficient to prove our main lemma.

\begin{lemma}\label{lem:decompcorrectness}
The process described above finds a generalized low depth tree decomposition of original tree $T$ of height $h\in O(\log^2 n)$ in 1 round with $O(n^\epsilon)$ local memory and $O(n\log^2n)$ total space.
\end{lemma}

\section{Calculating the smallest singleton cut
}\label{sec:singleton}

In this section, we show a $O(1/\epsilon)$ round AMPC algorithm that executes a series of contractions and outputs the size of the smallest singleton cut that appeared during the contraction process. That is we prove the following result.

\begin{theorem}\label{thm:smallest-singleton}
There exists an AMPC algorithm that given a graph $G$ with unique weights on edges calculates the minimum singleton cut that appears during the contraction process in $O(1/\epsilon)$ rounds using $O(n^{\epsilon})$ local memory and $O((n + m) \log^2{n})$ total space.
\end{theorem}

\subsection{Contraction process}
We view the contraction process of a weighted graph $G = (V, E, w : E \rightarrow [n^3])$ as a sequential process in which we iterate over multiple timesteps $0$ to $n^3$. For a given time $i$, we contract the edge $e$ having $w(e) = i$ to a single vertex. Let $G_{0}, \ldots, G_{n^3}$ be the sequence of graphs created in the process, where $G_{0}$ denotes the graph before any contraction and $G_{n^3}$ denotes the graph after all contractions. Via a quick comparison to Kruskal's algorithm, it is clear that the edges whose contraction changed the topology of the graph must belong to the minimum spanning tree of the weighted graph $G$ (since weights are unique, the MST is unique as well). Let $T = (V, E_{T}, w : E_{T} \rightarrow [n^3])$ be the minimum spanning tree of $G$. 

From the previous observation, it is enough to consider only contracting edges from tree $T$, which we will focus on in the rest of this section. It will also be convenient visualize vertices as simply being grouped instead of fully contracted. 

\begin{definition}\label{def:bag}
A \textbf{bag} of vertex $v$ at \textit{time} $t \in [n^3]$, which we denote $\bag(v, t)$, is the set of vertices that can be reached from $v$ using only edges of tree $T$ of weight at most $t$. We denote $\nbag(v,t)$ for set of \textbf{neighbors of a bag}, that is set of these vertices $u$ that do not belong to the bag and there exists an edge connecting $u$ and any vertex of the bag of weight greater than $t$. The \textbf{degree of a bag}, denoted $\Delta\bag(v, t)$, is the size of the set $\nbag(v, t)$.
\end{definition}

If we proceed with our edge contraction process, where an edge with weight $t$ is contracted at time $t$, then $\bag(v,t)$ is the set of all vertices that have been contracted with $v$ at time $t$.
The value $\Delta \bag(v,t)$ is simply the degree of the vertex that corresponds to contracted vertices. Therefore, the following simple observations holds.

\begin{observation}\label{obs:bags-eq-contr}
The value of the minimum singleton cut in the contraction process of the weighted graph $G$ is equal to the following:
$$\min_{v \in V, t \in [n^3]} \Delta \bag(v, t).$$
\end{observation}

\subsection{Simulating tree contractions with low depth decomposition}

By Observation~\ref{obs:bags-eq-contr} our goal is to calculate the value of
$$\min_{v \in V, t \in [n^3]} \Delta \bag(v, t).$$
To find this, we could calculate the value $\min_{t \in [n^3]} \Delta \bag(v, t) $ for every vertex $v$ independently in parallel. However, this would require a minimum of $\Omega(n \cdot (n + m))$ total space, which roughly corresponds to replicating the whole graphs for each independent instance. There are two key observations that will allow us to reduce the space complexity. First, bags are determined solely from the topology of tree $T$. Second, for larger $t$, it is likely the case that $\bag(u, t) = \bag(v, t)$, so we would like to remove this redundant computation. Therefore, we will exploit tree properties and the low depth decomposition to partition the work and avoid redundancy.

Let $\ell : V \rightarrow [h], h \in O(\log^2 n)$ be the labeling from the generalized low depth decomposition of tree $T$ (see Definition~\ref{def:lowdepth}). Let us asses to each bag a uniquely chosen vertex.
\begin{definition}
The \textbf{leader} of a bag, denoted $\bagLeader(v, t)$, is the vertex $u$ with the smallest label $\ell(u)$ among all vertices from $\bag(v, t)$. We define a number $\ldrtime(v)$ to be the greatest number $0 \le t' \le n^3$ such that $\bagLeader(v, t') = v$. 
\end{definition}

Let us first argue the correctness of the above definitions.

\begin{lemma}\label{lem:ldrtime}
The leader of every bag can be determined uniquely. Also, for every vertex $v \in V$ it holds: the number $\ldrtime(v)$ exists,  $\ldrtime(v) \ge 0$, and for every $0 \le t' \le \ldrtime(v)$ we have that $\bagLeader(v , t') = v$.
\end{lemma}

Using the fact that each bag has exactly one leader, we can reformulate the expression $\min_{v \in V, t \in [n^3]}$ $\Delta \bag(v, t)$ as follows
$$\min_{v \in V, t \in [n^3]} \Delta \bag(v, t) = \min_{v \in V} \min_{0 \le t \le \ldrtime(v)} \Delta \bag(v, t).$$
We then will distribute the work needed to calculate the right-hand side of the above equality by requiring each vertex to calculate the minimal degree among bags for which it is the leader:
$$\min_{0 \le t \le \ldrtime(v)} \Delta \bag(v, t).$$

Let $i$ be a number in $\left[\ceil{log^{2}n}\right]$. Let $L_i$ (the $i$th level) be the set of vertices $v\in V$ with low depth decomposition label $\ell(v) = i$, and $L_{\leq i}$ be that with label $\ell(v) \leq i$ (for convenience we assume that $ L_{\le 0} = \emptyset$). Let $T^{i}$ be the tree $T$ with $L_{\le i - 1}$ removed.  
The following observation, derived from the fact that a bag is a connected subgraph of $T$ and the leader has lowest value $\ell(\cdot)$, relates bag location to the topology of the low depth decomposition.

\begin{observation}\label{obs:disjoint}
For every $i \in \left[\ceil{log^{2}n}\right]$, $v \in L_{i}$, and $0 \le t \le \ldrtime(v)$, the set $\bag(v, t)$ belongs to a single connected component of graph $T^{i}$. For any two $u, v \in L_{i}$, sets $\bag(u, \ldrtime(u))$ and $\bag(v, \ldrtime(v))$ belong to different components of graph $T^{i}$.
\end{observation}

Recall, that we wanted to calculate the value
$$\min_{0 \le t \le \ldrtime(v)} \Delta \bag(v, t)$$ 
for every $v \in V$, which we rewrote as 

$$\min_{v \in V} \min_{0 \le t \le \ldrtime(v)} \Delta \bag(v, t).$$

Grouping by vertices in the same layers, we get

\begin{align*}
&\min_{v \in V} \min_{0 \le t \le \ldrtime(v)} \Delta \bag(v, t) \\
&\qquad \qquad = \min_{i \in \left[\ceil{log^{2}n}\right]} \min_{v \in L_{i}} \min_{0 \le t \le \ldrtime(v)} \Delta \bag(v, t).
\end{align*}

By Observations~\ref{obs:disjoint}, we can hope that computing the value
$$\min_{v \in L_{i}} \min_{0 \le t \le \ldrtime(v)} \Delta \bag(v, t),$$

can be done in parallel without exceeding global memory limit of $O(m \log^2 n)$, since for different $v \in L_{i}$, their bags up to time $\ldrtime(v)$ belong to different components of $T^{i}$, thus we might avoid redundant work. The details of computing this value are presented in the next section. Let us now formalize the progress so far.
\begin{lemma}\label{lem:to-level-reduction}
Given a tree $T$ and a graph $G = (V, E, w : E \rightarrow [n^3])$ as an input, calculating the value $\min_{v \in V, t \in [n^3]} \Delta \bag(v, t)$
can be reduced to $O(\log^{2} n)$ instances of calculating values
$$\min_{v \in L_{i}} \min_{0 \le t \le \ldrtime(v)} \Delta \bag(v, t),$$
for $i \in \left[\ceil{log^{2}n}\right]$. 
The reduction can be implemented in AMPC with $O(1 / \epsilon)$ rounds, $O((n + m) \log^2{n})$ total space, and $O(n^{\epsilon})$ local memory.
\end{lemma}
\begin{proof}
The correctness follows from the above discussion. 
For the implementation, the generalized low depth decomposition of $T$ can be determined in $O(1 / \epsilon)$ rounds with $O(n \log^{2} n)$ total space by Lemma~\ref{lem:treedecomp}. Consider now $O(\log^{2})$ tuples of format $(T, \ell, E, L_{i})$. Preparing them requires $O((n + m) \log^2 n)$ total space and the above discussion shows that the value 
$$\min_{v \in L_{i}} \min_{0 \le t \le \ldrtime(v)} \Delta \bag(v, t)$$
for every $i \in \left[\ceil{log^{2}n}\right]$ can be computed from the tuple $(T, \ell, E, L_{i})$, thus the lemma follows.
\end{proof}

\subsection{Resolving the problem for vertices on the same level.}

Following Lemma~\ref{lem:to-level-reduction}, we fix $i \in \left[\ceil{log^{2}n}\right]$ and set $L_{i}$. We calculate:
$$\min_{v \in L_{i}} \min_{0 \le t \le \ldrtime(v)} \Delta \bag(v, t).$$ 
In this approach, we will frequently query the minimum value over a path in a tree, thus the following result is helpful. 
\begin{theorem}[Behnezhad et al. \cite{behnezhad2019exponentially}]\label{thm:heavy-light-rmq}
Consider a rooted, weighted tree $T$, the heavy-light decomposition of this tree together with an RMQ data structure that supports queries on heavy paths can be computed in $O(1 / \epsilon)$ AMPC rounds using $O(n^{\epsilon})$ local memory and $O(n \log n)$ total space. If the aforementioned data structures are precomputed, then obtaining a minimum value on a path of a tree can be calculated with $O(\log n)$ queries to global memory.
\end{theorem}
We will also make use of the following theorem.
\begin{theorem}[Behnezhad et al. \cite{behnezhad2019near}]\label{thm:prefix-sums}
For a given sequence of integer numbers $S$ of length $n$, computing the minimum prefix sum over all prefix sums can be done in $O(1 / \epsilon)$ AMPC rounds using $O(n^{\epsilon})$ local memory and $O(n \log n)$ total space.
\end{theorem}

Finally, we show that the construction of the low depth decomposition provided in Section~\ref{sec:low-depth} gives easy access to edges that connect vertices of higher labels with vertices of smaller labels.
\begin{lemma}\label{lem:two-edges-decom}
For any connected component $C^i$ in $T^i$, there are at most $2$ tree edges between $C^i$ and $V \setminus T^i$ according to the low depth decomposition $\ell$ given in Lemma~\ref{lem:decompcorrectness}. Moreover, both edges can be calculated in $O(1 / \epsilon)$ AMPC rounds with $O(n^{\epsilon})$ memory per machine and $O(n \log^2 n)$ total memory.
\end{lemma}


Let us now turn to the proper part of this subsection.
First, we show how to compute values $\ldrtime(v)$ for all $v \in L_{i}$.
\begin{lemma}\label{lem:ldr-computation}
Given a tuple $(T, \ell, E, L_{i})$ for tree $T$, low depth decomposition $\ell$, set of weighted edges $E$, and levels $L_i$ for some $i\in [\lceil\log^2n\rceil]$, there exists an AMPC algorithm that calculates the value $\ldrtime(v)$
for every $v \in L_{i}$, in $O(1/\epsilon)$ rounds using $O(n^{\epsilon})$ local memory and $O((n + m) \log^2{n})$ global memory.
\end{lemma}
\begin{proof}
Consider vertex $v \in L_{i}$. Vertex $v$ ceases to be the leader of a bag at the first time $t$ when its bag is contracted with another bag containing at least one vertex of the set $L_{\le i - 1}$. According to the tree contraction process, time $t$ is equal to the largest weight of tree edges between $v$'s connected component in graph $T^{i}$ and the set of vertices $L_{\le i- 1}$. By Lemma~\ref{lem:two-edges-decom}, these edges can be extracted with at most $O(\log^2{n})$ queries to the low depth decomposition structure. We then simply find the minimum. Thus, all values $\ldrtime(v)$ for vertices from $L_{i}$ can be computed in constant number of rounds assumed the conditions stated in the lemma.
\end{proof}

\ignore{
We claim that $v$'s connected component in graph $T^{i}$, has at most two edges connecting it with set of vertices $L_{\le i - 1}$.  
The time of this contraction can be defined as $\min \{ w (e): e \in E_T(U_v, T \setminus U_v) \} $, where $E_T(U_v, T \setminus U_v)$ is the set of edges of $T$ between sets $U_{v}$ and $T \setminus U_v$. Now, the key observation is that the due to construction of the generalized low depth decomposition the set $E_T(U_v, T \setminus U_v)$ is of size at most $2$. To see this, we need to recall the structure of generalized low depth decomposition. It consists of $\log{n}$ binarized paths. If a vertex is the vertex with the lowest label on such a path, then the only edge  belonging to $E_T(U_v, T \setminus U_v)$ is the one that connects its binarized path to the path with the vertices on the lower level. 
On the other hand, if the vertex's label is not the lowest among other vertices from its path, then in the set $E_T(U_v, T \setminus U_v)$ there may be an additional edge connecting this vertex to a neighbor belonging to the same binarized path but with a lower label.

Each of those two edges can be find in $O(\log^{2}(n))$ sequential queries to the memory where the generalized low depth decomposition is stored.
This proves that for any $v \in L_{i}$ the value $\ldrtime(v)$ can be found in a single round in AMPC model. 
}

We can assume that values $\ldrtime(v) \in L_{i}$ are known. We would like to efficiently compute
$$\min_{0 \le t \le \ldrtime(v)} \Delta \bag(v, t),$$
for each $v \in L_{i}$.
For this, we make the following observation.

\begin{lemma}\label{lem:time-intervals}
Consider an edge $(x,y) =: e \in E$ and a vertex $v \in L_{i}$. All possible values $0 \le t' \le \ldrtime(v)$ at which $e$ belongs to set $\nbag(v, t')$ form a consecutive (possible empty) interval of integers $[a_{e}, b_{e}] \subseteq [0, \ldots, \ldrtime(v)]$, called also a $\textbf{time interval}$ with respect to $v$.
\end{lemma}

\begin{proof}
The lemma follows immediately from the fact that\linebreak $\bag(v, 0) \subseteq \bag(v, 1) \subseteq \ldots \subseteq \bag(v, n^3)$.
\end{proof}

\begin{figure}
    \centering
    \scalebox{0.7}{\tikzset{every picture/.style={line width=0.75pt}} 

\begin{tikzpicture}[x=0.75pt,y=0.75pt,yscale=-0.8,xscale=0.8]

\draw  [draw opacity=0][fill={rgb, 255:red, 74; green, 144; blue, 226 }  ,fill opacity=1 ] (300.5,95.25) .. controls (300.5,87.38) and (306.88,81) .. (314.75,81) .. controls (322.62,81) and (329,87.38) .. (329,95.25) .. controls (329,103.12) and (322.62,109.5) .. (314.75,109.5) .. controls (306.88,109.5) and (300.5,103.12) .. (300.5,95.25) -- cycle ;
\draw  [draw opacity=0][fill={rgb, 255:red, 74; green, 144; blue, 226 }  ,fill opacity=1 ] (300.5,165.25) .. controls (300.5,157.38) and (306.88,151) .. (314.75,151) .. controls (322.62,151) and (329,157.38) .. (329,165.25) .. controls (329,173.12) and (322.62,179.5) .. (314.75,179.5) .. controls (306.88,179.5) and (300.5,173.12) .. (300.5,165.25) -- cycle ;
\draw  [draw opacity=0][fill={rgb, 255:red, 74; green, 144; blue, 226 }  ,fill opacity=1 ] (300.5,235.25) .. controls (300.5,227.38) and (306.88,221) .. (314.75,221) .. controls (322.62,221) and (329,227.38) .. (329,235.25) .. controls (329,243.12) and (322.62,249.5) .. (314.75,249.5) .. controls (306.88,249.5) and (300.5,243.12) .. (300.5,235.25) -- cycle ;
\draw  [draw opacity=0][fill={rgb, 255:red, 74; green, 144; blue, 226 }  ,fill opacity=1 ] (300.5,305.25) .. controls (300.5,297.38) and (306.88,291) .. (314.75,291) .. controls (322.62,291) and (329,297.38) .. (329,305.25) .. controls (329,313.12) and (322.62,319.5) .. (314.75,319.5) .. controls (306.88,319.5) and (300.5,313.12) .. (300.5,305.25) -- cycle ;
\draw  [draw opacity=0][fill={rgb, 255:red, 126; green, 211; blue, 33 }  ,fill opacity=1 ] (370.5,235.25) .. controls (370.5,227.38) and (376.88,221) .. (384.75,221) .. controls (392.62,221) and (399,227.38) .. (399,235.25) .. controls (399,243.12) and (392.62,249.5) .. (384.75,249.5) .. controls (376.88,249.5) and (370.5,243.12) .. (370.5,235.25) -- cycle ;
\draw  [draw opacity=0][fill={rgb, 255:red, 126; green, 211; blue, 33 }  ,fill opacity=1 ] (370.5,305.25) .. controls (370.5,297.38) and (376.88,291) .. (384.75,291) .. controls (392.62,291) and (399,297.38) .. (399,305.25) .. controls (399,313.12) and (392.62,319.5) .. (384.75,319.5) .. controls (376.88,319.5) and (370.5,313.12) .. (370.5,305.25) -- cycle ;
\draw  [draw opacity=0][fill={rgb, 255:red, 126; green, 211; blue, 33 }  ,fill opacity=1 ] (230.5,236.25) .. controls (230.5,228.38) and (236.88,222) .. (244.75,222) .. controls (252.62,222) and (259,228.38) .. (259,236.25) .. controls (259,244.12) and (252.62,250.5) .. (244.75,250.5) .. controls (236.88,250.5) and (230.5,244.12) .. (230.5,236.25) -- cycle ;
\draw  [draw opacity=0][fill={rgb, 255:red, 126; green, 211; blue, 33 }  ,fill opacity=1 ] (230.5,305.25) .. controls (230.5,297.38) and (236.88,291) .. (244.75,291) .. controls (252.62,291) and (259,297.38) .. (259,305.25) .. controls (259,313.12) and (252.62,319.5) .. (244.75,319.5) .. controls (236.88,319.5) and (230.5,313.12) .. (230.5,305.25) -- cycle ;
\draw    (314.75,109.5) -- (314.75,151) ;
\draw    (300.5,165.25) -- (244.75,222) ;
\draw    (314.75,179.5) -- (314.75,221) ;
\draw    (329,165.25) -- (384.75,221) ;
\draw    (244.75,250.5) -- (244.75,291) ;
\draw    (314.75,249.5) -- (314.75,291) ;
\draw    (384.75,249.5) -- (384.75,291) ;
\draw  [draw opacity=0][fill={rgb, 255:red, 245; green, 166; blue, 35 }  ,fill opacity=1 ] (440.5,304.25) .. controls (440.5,296.38) and (446.88,290) .. (454.75,290) .. controls (462.62,290) and (469,296.38) .. (469,304.25) .. controls (469,312.12) and (462.62,318.5) .. (454.75,318.5) .. controls (446.88,318.5) and (440.5,312.12) .. (440.5,304.25) -- cycle ;
\draw    (399,235.25) -- (454.75,290) ;
\draw  [dash pattern={on 4.5pt off 4.5pt}]  (454.75,290) .. controls (518,193) and (440,77) .. (329,95.25) ;

\draw (310,88) node [anchor=north west][inner sep=0.75pt]  [color={rgb, 255:red, 255; green, 255; blue, 255 }  ,opacity=1 ]  {$3$};
\draw (310,157) node [anchor=north west][inner sep=0.75pt]  [color={rgb, 255:red, 255; green, 255; blue, 255 }  ,opacity=1 ]  {$2$};
\draw (310,228) node [anchor=north west][inner sep=0.75pt]  [color={rgb, 255:red, 255; green, 255; blue, 255 }  ,opacity=1 ]  {$1$};
\draw (310,298) node [anchor=north west][inner sep=0.75pt]  [color={rgb, 255:red, 255; green, 255; blue, 255 }  ,opacity=1 ]  {$2$};
\draw (380,228) node [anchor=north west][inner sep=0.75pt]  [color={rgb, 255:red, 255; green, 255; blue, 255 }  ,opacity=1 ]  {$4$};
\draw (380,297) node [anchor=north west][inner sep=0.75pt]  [color={rgb, 255:red, 255; green, 255; blue, 255 }  ,opacity=1 ]  {$5$};
\draw (239,228) node [anchor=north west][inner sep=0.75pt]  [color={rgb, 255:red, 255; green, 255; blue, 255 }  ,opacity=1 ]  {$4$};
\draw (239,297) node [anchor=north west][inner sep=0.75pt]  [color={rgb, 255:red, 255; green, 255; blue, 255 }  ,opacity=1 ]  {$5$};
\draw (449,297) node [anchor=north west][inner sep=0.75pt]  [color={rgb, 255:red, 255; green, 255; blue, 255 }  ,opacity=1 ]  {$6$};
\draw (393,200) node [anchor=north west][inner sep=0.75pt]  [font=\Large]  {$\mathbf{v}$};
\draw (329,258) node [anchor=north west][inner sep=0.75pt]  [font=\footnotesize]  {$6\mid [ 0,\underline{1}]$};
\draw (350,163) node [anchor=north west][inner sep=0.75pt]  [font=\footnotesize]  {$\underline{2} \mid [ 0,1]$};
\draw (417,229) node [anchor=north west][inner sep=0.75pt]  [font=\footnotesize]  {$\underline{1} \mid [ 0,0]$};
\draw (468,128) node [anchor=north west][inner sep=0.75pt]  [font=\footnotesize]  {$\underline{12} \mid [ 0,1]$};

\end{tikzpicture}}
    \caption{A sample structure of an MST tree. Firm edges are \emph{tree} edges, while dotted are \emph{non-tree} edges. The number inside vertices denote their levels. Different colors symbolize different binarized paths. The numbers underlined are times of contraction of corresponding edges. Next to these number the time intervals of these edges with respect to vertex $v$ are given. Since $\ldrtime(v) = 2$, thus all these intervals are contained in $[0,2]$.}
    \label{fig:intervals}
\end{figure}
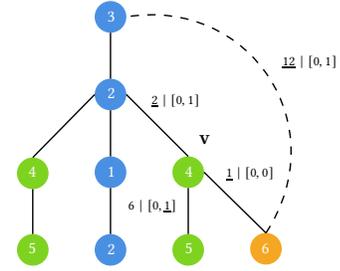

Additionally, the following observation shows, given edge time intervals, how to derive $\min_{0 \le t \le \ldrtime(v)} \Delta \bag(v, t)$ and clarifies the purpose of time intervals.

\begin{observation}\label{obs:intervals}
Fix a vertex $v \in V$ and consider time intervals $[a_{e}, b_{e}]$ with respect to $v$, for all $e \in E$. Denote the set of all intervals containing value $x$ by $I_x$. 
Then, computing the value
$$\min_{0 \le t \le \ldrtime(v)} \Delta \bag(v, t),$$
is equivalent to computing the minimum over all values $|I_x|$ for $x$ in the range $[0, \ldrtime(v)]$.
\end{observation}


Since this task is `linear', it can be computed efficiently in AMPC.  We now discuss how to compute the intervals for all edges in $E$.

\begin{lemma}\label{lem:intervals-computation}
Given a tuple $(T, \ell, E, L_{i})$, there exists an AMPC algorithm that for every vertex $v \in L_{i}$ and every edge $e \in E$ calculates the maximal, non-empty time interval $[e_{a}, e_{b}] \subseteq [0, \ldots, \ldrtime(v)]$ of $e$ with respect to $v$.
The algorithm works in $O(1/\epsilon)$ rounds, uses $O(n^{\epsilon})$ local memory and $O((n + m) \log^2{n})$ global memory.
\end{lemma}

\begin{proof}
The algorithm starts by removing vertices $L_{\le i - 1}$ with all edges adjacent to them from tree $T$ which gives us $T^{i}$. Given decomposition $\ell$, this can be done in $O(1)$ rounds. By definition~\ref{def:lowdepth}, vertices $L_{i} = \{v_{1}, \ldots, v_{q}\}$ belong to different  trees. Next, the algorithm roots these trees that contain vertices from $L_{i}$ in $v_{1}, \ldots, v_{q}$ and calculates heavy-light decompositions of each  tree together with an RMQ structure on heavy paths. By Theorem~\ref{thm:heavy-light-rmq}, this can be done in $O(1 / \epsilon)$ within our memory constraints.

Let us now fix an edge $(x, y) =: e \in E$. Importantly, we consider here all edges of the graph $G$, not only tree edges $E_{T}$. 
Let $r_{x} \in \{\perp, v_{1}, \ldots, v_q\}$ be the root of this tree in $T^{i}$ to which the vertex $x$ belongs. 
If the vertex $x$ does not belong to any tree, that is $x \in L_{\le i - 1}$ since these are the vertices that have been removed, we write $r_{x} = \perp$.
Let $\mw(x)$ be the minimum weight over edges of path that connects vertex $x$ with vertex $r_{x}$ in graph $T^{i}$. Observe, that unless $r_{x} = \perp$ this value is well defined as $T^{i}$ is a collection of tree and there is exactly one path connecting is $x$ and $r_{x}$ in this graph. We extend the above definitions on $y$ in the natural way.

By Theorem~\ref{thm:heavy-light-rmq}, computing $r_{x}, r_{y}, \mw(x), \mw(y)$ takes $O(\log n)$ queries to the memory for a single edge. Therefore, we can compute these values for all edges $e \in E$ in $O(1)$ round under the conditions assumed in this lemma.

Observe that edge $e = (x, y)$ can have non-empty time intervals only with vertices $r_{x}$ and $r_{y}$. Any other vertex from $L_{i}$ belongs to a different connected component in graph $T_{i}$ and therefore its bag cannot contain $x$ nor $y$ while the vertex is the leader of its bag. Thus, all that is left to show is how $\mw(x)$ and $\mw(y)$ can help determine the time intervals in which edge $e$ belongs to $\nbag(r_{x})$ and $\nbag(r_{y})$. 
We consider the following cases.

\noindent \textit{Case $1.$} $r_{x} = \perp, r_{y} = \perp$. In this case, edge $(x,y)$ has no effect on degrees of bags of vertices $r_{x}$ and $r_{y}$ at any time. The algorithm skips such edges.
    
\noindent \textit{Case $2.$} $r_{x} = \perp, r_{y} \neq \perp$, (or symmetrically $r_{x} \neq \perp, r_{y} = \perp$). Since $T^{i}$ is a subset of the minimum spanning tree $T$, thus the first time when vertex $x$ belongs to $r_{x}$'s bag is the time $\mw(x)$. Now, $y$ starts to belong to $r_{x}$'s bag either at the time being equal to the maximal weight on the path between $r_{x}$ and $y$. Observer however, that this path has to contain vertices that does not belong to $T^{i}$ and therefore the maximal weight has to be greater than $\ldrtime(r_{x})$. What follows the correct interval in this case is:
$$[\mw(x), \ldrtime(r_{x})],$$
or an empty interval if $\mw(x) >  \ldrtime(r_{x})$.
    
\noindent \textit{Case $3.$} $r_{x} \neq \perp, r_{y} \neq \perp$. We distinguish two sub-cases:

\textit{Subcase $a)$} $r_{x} \neq r_{y}$. Since the path between $r_{x}$ and $r_{y}$ does not belong to $T^{i}$ we can proceed analogously to the \textit{Case} $2.$. The correct interval for vertex $x$ is
$$[\mw(x), \ldrtime(r_{x})],$$
or an empty interval if $\mw(x) >  \ldrtime(r_{x})$, while for vertex $y$ it is
$$[\mw(y), \ldrtime(r_{y})],$$
or an empty interval if $\mw(y) >  \ldrtime(r_{y})$

\textit{Subcase $b)$} $r_{x} = r_{y}$. Since $T^{i}$ is a subgraph of the minimum spanning tree $T$, we have that $\min(\mw(x), \mw(y))$ is the first time when at least one of $x$ and $y$ belongs to $r_{x}$'s bag, while the first time when both belong to $r_{x}$'s bag is $\max(\mw(x), \mw(y))$. Thus, the proper time interval for this edge:

\begin{align*}
[\min(\mw(x), \mw(y)),& \max(\mw(x), \mw(y))]
\\&\cap [1, \ldots, \ldrtime(r_{x})]
\end{align*}

We obtain that for every edge $e \in E$ all non-empty time intervals in which this edge belong to $\nbag$ of some vertex $v$ can be computed in $O(\log(n))$ queries to the memory. Therefore, computing these values for all edges can be done in constant number of rounds assumed $O(m \log n)$ total memory.
\end{proof}

Implementing Observation~\ref{obs:intervals} is purely technical.
\begin{lemma}\label{lem:intervals-imp}
There exist an AMPC algorithm that given a set of integer intervals $I = \{ [p_{1}, k_{1}], \ldots, [p_{n}, k_{n}] \}$, $\forall_{i \in [n]} [p_{i}, k_{i}] \subseteq [0, R]$ finds the minimal number of intersecting intervals in $O(1 / \epsilon)$ rounds using $O(n^{\epsilon})$ local memory and $O(n \log^2 n)$ total memory.
\end{lemma}
\begin{proof}
First, the algorithm sorts the set $\{p_{1}, k_{1}, \ldots, p_{n}, k_{n} \}$ of all endpoints of these intervals in non-increasing order (ties are resolved with priority for endpoints $p_{i}$) obtaining a sequence $S$.
Consider assigning to every endpoint $p_{i}, i \in [n]$ from sequence $S$ value $+1$ and to every endpoint $k_{i}, i \in [n]$ value $-1$. This operation leads to a sequence $S'$ of pairs of format $($endpoint, value$)$.  Finally, let $S''$ be a sequence constructed from $S'$ in which all consecutive pairs that have the same first coordinate are compressed to a single pair in which the first coordinate is preserved and the second is the sum of second coordinates of contracted pairs. It can be observed that finding the minimal prefix sum of sequence made from second coordinates of pairs in $S''$ is equivalent to the minimal number of intersecting intervals. The construction of sequence $S''$ requires only sorting and contracting consecutive pairs which can be implemented in $O(1 / \epsilon)$ rounds in AMPC with the memory constrains stated in the lemma. To find the minimal prefix sum we use Theorem~\ref{thm:prefix-sums} which completes the proof.
\end{proof}
The above discussion is summarized in the following lemma.
\begin{lemma}\label{lem:bag-n-calculating}
There exists an AMPC algorithm that given a tuple $(T, \ell, E, L_{i})$ calculates the value
$$\min_{v \in L_{i}} \min_{0 \le t \le \ldrtime(v)} \Delta \bag(v, t)$$
in $O(1 / \epsilon)$ rounds using $O(n^{\epsilon})$ local memory and $O((n + m) \log^2 n)$ total memory.
\end{lemma} 
\begin{proof}
Using Lemma~\ref{lem:ldr-computation} we are able to calculate value $\ldrtime$ for every $v \in L_{i}$ in constant number of rounds. By Lemma~\ref{lem:intervals-computation} we can calculate time all non-empty time intervals for every $e \in E$ and every $v \in L_{i}$. This requires $O(m \log^{2} n)$ total memory. Each time interval $[a, b]$ can be assigned a vertex $v$ with respect to whom it was calculated. Then, we group time intervals with respect to vertices from $L_{i}$ they were calculated.
This can be done in a single round with $O(m \log^2 n)$ global memory since there are only $O(m)$ non-empty time intervals. 
Finally, Lemma~\ref{lem:intervals-imp} guarantees that we can compute, for every $v \in L_{i}$, the minimum number of intersecting intervals in $O(1 / \epsilon)$ rounds with total memory proportional to the number of these intervals. Therefore, assumed $O(m \log^{2} n)$ global memory we can extend the last computation to a parallel computation for $v \in L_{i}$ while preserving the round complexity. By Observation~\ref{obs:intervals} this is equivalent to calculating
$$ \min_{0 \le t \le \ldrtime(v)} \Delta \bag(v, t),$$
for every $v \in L_{i}$. Since the minimum of the above values over $v \in L_{i}$ can be computed in a single round, the lemma is proven.
\end{proof}

\ignore{
\begin{algorithm}[ht]
\SetAlgoLined
\KwData{A tuple $(T, D_{T}, E, L_{i})$.}
\KwResult{$\min_{v \in L_{i}} \min_{0 \le t \le \ldrtime(v)} \Delta \bag(v, t)$}
\ForEach{edge $(x,y) \in E(G)$, in parallel,}{
    Compute values $ldr(x)$, $ldr(y)$, $\ldrtime(ldr(x))$, $\ldrtime(ldr(x))$ and values $mw(x)$, $mw(y)$\;
    Based on these values compute time interval $[a_{x}, b_{x}]$ ($[a_{x}, b_{x}]$ respectively) that determines the first and the last moment when edge $(x,y)$ becomes outer edge to bag in which $ldr(x)$ is the leader vertex (or vertex $ldr(y)$ is the leader respectively)\;
}
Sort in ascending order all endpoints of the time intervals computed in the previous step\;

\ForEach{vertex $v_{j} \in L_{i}$, in parallel,}{
Let $I_{v_{j}}$ be the set of time intervals that correspond to edges incident to $v_{j}$'s connected component in graph $T \setminus L_{i}$\;
Treated a beginning of an interval from $I_{v_{j}}$ as $+1$ and an end of an interval from $I_{v_{j}}$ as $-1$, compute the minimum prefix sum over the sequence of sorted endpoints using Behnezhad's algorithm~\cite{behnezhad2020parallel}\;
Assign the computed minimum $MSC(v_{j})$\;
}
\KwRet{values $MSC(v_{j})$ for each $v_{j} \in L_{i}$}\;
\caption{\textsf{SingeltionCutsSameLevel} \\(An algorithm that computes ..)}
\label{alg:smallestsingletoncutsingle}
\end{algorithm}
}

\subsection{The final algorithm.}
We are now able to prove Theorem~\ref{thm:smallest-singleton} and present the final algorithm, \textsf{SmallestSingletonCut}, that calculates the smallest singleton cut that appears in the contraction process of $G$. The pseudcode can be found in Figure~\ref{alg:smallestsingletoncut}, while the proof of correctness is below.


\begin{algorithm}
\SetAlgoLined
\KwData{Graph $G = (V, E, w : V \rightarrow [n^3])$.}
\KwResult{Size of the smallest singleton cut.}
Compute the minimum spanning tree $T$ of $G$\; \label{line:spanning-tree}
Compute the low depth decomposition $D_{T}$ of $T$\;
Prepare $O(\log^{2} n)$ tuples $(T, D_{T}, E, L_{i}), i \in \left[\ceil{log^{2}n}\right]$\; \label{line:tuples-b}
\ForEach{tuple $(T, D_{T}, E, L_{i})$}{
    Compute: $\mathsf{lc}_i \leftarrow \min_{v \in L_{i}} \min_{0 \le t \le \ldrtime(v)} \Delta \bag(v, t)$\;
}

\KwRet{$\min(\mathsf{lc}_1, \ldots, \mathsf{lc}_{\left[\ceil{log^{2}n}\right]})$}\; \label{line:tuples-e}
\caption{\textsf{SmallestSingletonCut}}
\label{alg:smallestsingletoncut}
\end{algorithm}

\begin{proof}[Proof of Theorem~\ref{thm:smallest-singleton}]
The correctness follows from Observation~\ref{obs:bags-eq-contr} and Lemmas~\ref{lem:to-level-reduction} and~\ref{lem:bag-n-calculating}.
Also the implementation details of lines $\ref{line:tuples-b}-\ref{line:tuples-e}$  are discussed in the these two lemmas. To calculate minimum spanning tree in line~\ref{line:spanning-tree} we use Lemma~\ref{lem:orient} while the implementation of the low depth decomposition from the next line is given by Lemma~\ref{lem:treedecomp}. 
\end{proof}

\section{AMPC algorithm for approximated minimum $k$-cut}\label{k-min-cut}
In this section, we show that given an algorithm that calculates $2+\epsilon$ approximation of a min cut, one can construct $4 + \epsilon$ approximation of minimum k-cut.
Consider the following greedy algorithm, called \textsf{APX-SPLIT}, that extends the classic result of Saran and Vazirani~\cite{SaranV95}. The algorithm works sequentially. In each iteration, it extends the approximated solution with the smallest non-trivial approximation of the minimum cut of the graph available at a given moment. Being precise, assume that at the beginning $i$-th iteration the algorithm has split the graph $G$ into $\ell_{i}$ connected components (we start with $c_{1} = 1$). Then the algorithm calculates $(2+ \epsilon)$-approximation of the minimum cut in each connected component and enlarges the solution by the smallest of these cuts, thereby increasing the number of components by at least $1$. The algorithm ends after the first iteration such that the number of connected components after this iteration is at least $k$. The pseudocode of the algorithm is given in the Algorithm~\ref{alg:apx-split}.

\begin{algorithm}[ht]
\SetAlgoLined
\KwData{A graph $G = (V(G), E(G))$ and a parameter $k$.}
\KwResult{A $(4 + \epsilon)$-approximation of minimum $k$-cut.}
Initialize $D \to \emptyset$\;
\While{$G' := (V(G), E(G) \setminus \bigcup_{d \in D} d)$ has less than $k$ connected components} {
    Let $C_{1}, \ldots, C_{l}$ be the set connected of components of $G'$\;
    $d_i^* \gets \textsf{AMPC-MinCut}(C_i)$ for all $i\in [l]$\;
    $j \gets \arg\min_{i\in[l]} weight(d_i)$\;
    Add $d_{j}^*$ to $D$\;
  }
  \KwRet{set of cuts $D$}\;
 \caption{\textsf{APX-SPLIT} \\(A greedy algorithm computing an approximation of the minimum $k$-cut in AMPC)}
\label{alg:apx-split}
\end{algorithm}
 
We will show by generalizing the idea of Saran and Vazirani that the aforementioned greedy algorithm is $(4+ \epsilon)$-approximation minimum $k$-cut. 

\begin{customthm}{\ref{thm:vazirani'scuts}}
\thmvazirani
\end{customthm}

We define $comps (c_{1} \cup \ldots \cup c_{k}) $ the number of components of the graph $G$ after removing all edges from the set $c_{1} \cup \ldots \cup c_{k}$. 

\begin{proof}

The standard line of proof, proposed in~\cite{SaranV95} for the case whe exact minimum cut is used at each splitting step, is to compare the cut selected by the \textsf{APX-SPLIT} algorithm to the approximated minimum $k$-cut obtained from the Gomory-Hu tree. The main difficulty is that in our case, we use \textit{only} $(2+\epsilon)$ approximation in each splitting step. This makes our proof different and novel compared to~\cite{SaranV95}.
Let us set a Gomory-Hu tree $H = (V(H) = V(G), E(G))$ of the graph $G$. The Gomory-Hu tree is defined as follows.
\begin{definition}[Gomory and Hu~\cite{GomoryHu61trees}]
Consider an arbitrary graph $G$.
A weighted tree $H = (V(H), E(H))$ with the set of vertices being equal $V(G)$ is called a Gomory-Hu tree of $G$, if for every pair of different vertices $s, t \in V(G)$ the minimum weight on the path between $s$ and $t$ in the tree $H$ is equal to the minimum $s$-$t$ cut in graph $G$. The existence and construction of Gomory-Hu trees was shown in~\cite{GomoryHu61trees}.
\end{definition}

Let us order edges of the tree $H$ (or equivalently cuts in the $G$ graph) with respect to non-decreasing weights. Denote $b^*_{1}, \ldots, b^*_{l-1}$ the sequence of the first $l \le k-1$ edges (cuts equivalently) from this order such that corresponding cuts split $G$ graph into at least $k$ connected components. Let $b_{1}, \ldots, b_{k-1}$ be a dual sequence of cuts corresponding to that sequence of edges $b^*_{1}, \ldots, b^*_{l-1}$, with this addition that we put each cut $b_{i}$ this number of times it increases the number of connected components in $G$ and possibly cut some suffix of such generated sequence to obtain exactly $k-1$ cuts. For such construction we have the following.
\begin{observation}[Saran and Vazirani~\cite{SaranV95}]\label{obs:comps}
The sequence of cuts $b_{1}, \ldots, b_{k-1}$ satisfies:
\begin{enumerate}
    \item the sequence $|b_{1}|, \ldots, |b_{k-1}|$ is non-decreasing,
    \item $\forall_{i \in [k-1]} comps(b_{1} \cup \ldots \cup b_{i} ) > i.$
\end{enumerate}
\end{observation}
Saran and Vazirani also proved that such selected (and possibly refactored) sequence of cuts is $(2-\frac{2}{k})$-approximation of the minimum $k$-cut. 
\begin{theorem}[Saran and Vazirani~\cite{SaranV95}]
The cut $\bigcup_{i \in [k-1]} b_{i}$ is $(2-\frac{2}{k})$ approximation of the minimum $k$-cut of $G$.
\end{theorem}
Having established the crucial properties of Gomory-Hu trees and corresponding cuts we can proceed to the proof of correctness of the \textsf{APX-SPLIT} algorithm.

Let $d_{1}, \ldots, d_{m}$ be the successive cuts selected by the \textsf{APX-SPLIT} algorithm. Note that with each new cut, at least one new component appears in the graph thus $m \le k - 1$.
We will show that the sum of these cuts' sizes is not greater than the sum of sizes of cuts $b_{1}, \ldots, b_{k - 1}$. Let $\#c_{1}, \ldots, \#c_{m}$ be a sequence of numbers where $\#c_{i} := comps(d_{1} \cup \ldots \cup d_{i})$.
We will show by induction that 
$$\forall_{j \in [m]} |d_{1} \cup \ldots d_{j}| \le (2 + \epsilon)|b_{1} \cup \ldots \cup b_{\min(k - 1, \#c_{j})}|,$$.

The idea behind the induction step defined in the previous line can be explained as follows: inclusion of cuts from $d_1$ to $d_i$ are at least $(2 + \epsilon)$ approximation of cut generated by inclusion of cuts of such prefix of sequence $b_{1}, \ldots, b_{k}$ that split $G$ on $\#c_{i}$ for connected components.

For the basis of induction we see that in the first step of the \textsf{APX-SPLIT} algorithm chooses $(2+ \epsilon)$-approximation of the smallest cut in the whole graph $G$. The $b_{1}$ is an $s$-$t$ cut therefore we have $|d_ {1}| \le (2+ \epsilon) |b_ {1}|$ which implies that $|d_ {1}| \le (2+ \epsilon) |b_ {1} \cup \ldots b_{\min(k - 1, \#c_{1})}|$. 

Now consider $i \in [m - 1]$. Since $i < m$, we observe that $\#c_{i} < k$. Otherwise the algorithm  \textsf{APX-SPLIT} would have executed only $m-1$ iterations instead of $m$. Consider cuts $b_{1}, \ldots, b_{\#c_i + 1}$. From the Observation~\ref{obs:comps} we have that $comps(b_{1} \cup \ldots \cup b_{\#c(i) + 1}) > \#c(i)$. On the other hand $comps(d_{1} \cup \ldots d_{i}) = \#c_{i}$. Since  $comps(b_{1} \cup \ldots \cup b_{\#c_{i} + 1}) > comps(d_{1} \cup \ldots d_{i})$ then there must be a cut $b_ {j}$, $j \in [\#c_{i} + 1]$ that is not covered by the first $i$ cuts from the sequence $d_{1}, \ldots, d_{m}$. Namely, we can choose $j$ such that $b_{j} \nsubseteq d_{1} \cup \ldots \cup d_{i}$. Moreover, since $b_{j}$ is an $s$-$t$ cut in the graph $G$ (with all edges included), thus it must split at least one connected component of the graph $G = (V, E \setminus (d_{1} \cup \ldots d_{i}))$ into two non-empty parts. Thus this cut is considered in the $i+1$-th iteration of the \textsf{APX-SPLIT} algorithm, which implies that $|d_{i + 1}| \le (2 + \epsilon)|b_{j}| \le (2 + \epsilon)|b_{\#c_{i} + 1}|$. 
Since $\#c_{i} + 1 \le \#c_{i + 1}$, we conclude that $|d_{1} \cup \ldots d_{i + 1}| \le (2 + \epsilon)|b_{1} \cup \ldots \cup b_{\min(k - 1, \#c_{i + 1})}|$, which proves the inductive step. 
Now, we see from Theorem~\ref{thm:vazirani'scuts} that the solution of $b_{1}, \ldots, b_{k - 1} $ is an $(2- \frac{2}{k})$-approximation of the minimum $k$-cut. Thus the solution $d_{1} \ldots d_{k-1}$ is $(2 + \epsilon)(2 - \frac{2}{k}) = $  approximation of the minimum $k$-cut. This proves the correctness of the algorithm.

It remains to be noted that a single iteration of the algorithm can be performed in $\bigO(\log\log n)$ rounds in the AMPC model with $\bigO(n^{\epsilon})$ memory per machine and in total memory $\bigO(m)$. The dominant operation is the calculation of a $(2+\epsilon)$ approximation of the minimum cut in each of the components. Its performance is analyzed in Theorem~\ref{thm:mincut}. The calculation of the smallest of all approximated cuts corresponding to different components is a standard operation and can be performed in $O(1)$ rounds. Also, Behnezhad et. al in~\cite{behnezhad2020parallel} showed that the number of components of a graph can be determined in $\bigO(1)$ rounds in AMPC with $\bigO(n^{\epsilon})$ memory per machine and $\bigO(m)$ total memory. This completes the performance analysis of the algorithm.
\end{proof}

\bibliographystyle{ACM-Reference-Format}
\bibliography{references}

\section{Missing Proofs from Section \ref{sec:low-depth}}

\begin{proof}[Proof of Lemma~\ref{lem:heavylight}]
First off, we know subtree size can be computed in $O(1/\epsilon)$ low-memory AMPC rounds on trees as shown by Behnezhad et al.~\cite{behnezhad2019massively}, and the child of a vertex $v$ with minimum subtree can then be found in $O(1/\epsilon)$ rounds by dividing the children amongst machines and iteratively finding the smallest. Next, consider removing all light edges from $T$ to create a forest $F$. Run Behnezhad et al.'s AMPC connectivity algorithm~\cite{behnezhad2020parallel}, which satisfies the round and space constraints, to identify the components and contract them. Add the light edges back in to connect contracted nodes. This clearly is $T_M$. Additionally, as before, we run Behnezhad et al.'s~\cite{behnezhad2019massively} AMPC algorithm for orienting the tree. This too falls within the constraints.
\end{proof}

\begin{proof}[Proof of Observation~\ref{obs:ancestors}]
Consider such a $u,u',u''\in P_v$ that appear in this order (or reversed), and let $v$ and $v'$ be the lowest common ancestors of $u$ and $u'$, and $u$ and $u''$ respectively. Consider traversing from $u$ up the tree from child to parent, and let $p$ be the current vertex. We start with $p=u$ and thus the leaf set of the subtree rooted at $p$ is $L_p = \{u\}$. As we traverse upwards, we add sets of leaves to $L_p$ that are contiguous in $P_v$. Additionally, one vertex is directly adjacent to a vertex from $L_p$ in $P_v$ because $P_v$ was mapped to the leaves of $T_v$ according to the pre-order traversal. Therefore, $L_p$ is a contiguous chunk of $P_v$. Thus, when $u''$ gets added to the subtree (i.e., when $p=v'$) $u'$ must be in $L_p$ too, either because it was added previously or it is being added at the same time. In the former case, $v$ must have happened before $v'$ and thus $v'$ is an ancestor of $v$, and in the latter case, $v=v'$.
\end{proof}

\begin{proof}[Proof of Lemma~\ref{lem:expandverts}]
Correctness of the process described in this section is seen in Observation~\ref{obs:binpathcorrect}. Thus we simply need to show how to implement it in AMPC. For the first step, we must construct a generic almost complete binary tree with $|P_v|$ leaves. Call this tree $T_v$. By Observation~\ref{obs:binarytree}, this has $2|P_v|-1$ vertices, $\lfloor\log_2|P_v|\rfloor+1$ max depth, and each layer is full except the last which has $2|P_v|-2^{\lfloor\log_2|P_v|\rfloor +1}$ vertices. Thus, it is fairly simple to, in parallel, create the set of all vertices in the tree and then connect each vertex to its parent. Each vertex can be given an index: a unique identifier for vertices numbered $1,\ldots,2|P_v|-1$. This is going to represent the order of the vertices in a breadth-first traversal of the tree. For a vertex with index $1<i< 2|P_v|$, its parent's index $j$ can be computed as $j=\lfloor i/2\rfloor$. Whether or not a vertex is a left or right child is simply determined by the parity of its pre-order index.

Note that the size of the tree is $O(|P_v|)$ but each individual processor computation (i.e., computing the size of the tree, and then having each index connect itself to its parent) can be done in constant space and 2 rounds. Thus in 2 rounds, we can create such a tree. Note that we have to construct these trees in parallel, but it is not hard to see that this will only require $O(n\log n)$ total space which can be divided appropriately amongst machines.

For the second step, we can use Behnezhad et al.'s~\cite{behnezhad2019massively} algorithm for pre-order numbering with $O(1/\epsilon)$ AMPC rounds w.h.p. using $O(n^\epsilon)$ local memory and $O(n)$ total space. Let $L$ be a list of the leaves of $T_v$ in pre-order. To map $P_v$ to the leaves, one can simply do a direct map between $P_v$ and $L$ in one round. 
\end{proof}

\begin{proof}[Proof of Lemma~\ref{lem:decompcorrectness}]
First, note that $P^u$ can be stored entirely on one machine by Observation~\ref{obs:depth}, and additionally, since both $T_M$ and $T_v$ for all $v\in M$ is oriented, it is quite simple to adaptively query the path from $u$ to the root in one round within the space constraints. Assuming the orientations also labels if the vertex is a left or right child, $u'$ can be found simply by searching the path. Finally, the depth of $u'$, which is the label of $u$, can also be found quite simply given access to all of $P^u$.

It is quite simple to show the height is bounded by $O(\log^2n)$: all labels are depths in the expanded $T_M$ and Observation~\ref{obs:depth} bounds the max depth. All that is left is to show the connected components induced on $T^i = \{v\in T:\ell(v) \geq i\}$ contain exactly one vertex with label $i$ each.

We show this by induction on the level. At the 1st level, we should only partition the graph once. Let $u$ be a vertex labeled 1 and $P^u$ be its path. For $u$ to be given depth 1, it must have received its label from the root $r_M$ of the expanded $T_M$, since we are counting depth starting at 1. Thus, it must be the leftmost descendant of the right child of $r_M$. This is clearly unique, thus $u$ is unique. Therefore, there is exactly one vertex at the 1st level.

Consider a connected component $C$ in $T^i=\{v\in T:\ell(v) \geq i\}$ for some level $i\in h$. Let the ``neighborhood'' $N(C)$ be all the vertices in $T\setminus C$ that are adjacent to some vertex in $C$. We will first show that for the largest level $j$ of a vertex in $N(C)$, there is exactly one vertex in $N(C)$ of level $j$. Note $j < i$, otherwise a vertex of level $j$ in $N(C)$ would actually be included in the component $C$.

Assume for contradiction there are at least two such vertices, $u,v\in N(C)$ with level $j$. Let $u'$ and $v'$ be their respective neighbors in $C$. Since $C$ is a connected component, there is a path $P$ from $u'$ to $v'$ containing only vertices in $C$. All of these vertices have level $i$ or higher by the definition of $C$. Tacking on $u$ and $v$ to the start and end of $P$ respectively, there is a path from $u$ to $v$ consisting of vertices $\{u,v\}\cup P$. Since $u$ and $v$ are at level $j$, that means every vertex in this path has level $j$ or higher. Therefore, $u$ and $v$ must have been in the same connected component $C'$ in the earlier level $T^{j}$. By induction, that component must have had only one vertex at level $j$. This is a contradiction. Thus $N(C)$ must have exactly one vertex in level $j$.

Let $v\in N(C)$ be the vertex at level $j$, and let $u$ be its neighbor in $C$ (note there can only be one since $T$ is a tree). We consider three cases.

\paragraph{Case 1} $u$ is a child of $v$ and they are not on the same heavy path. Let $u_M$ and $v_M$ be the meta vertices in $T_M$ containing $u$ and $v$ respectively, $u'$ and $v'$ be the corresponding nodes found in step 2 for $u$ and $v$ (whose depths are the labels of $u$ and $v$), and $r_{u_M}$ be the root of the binarized path for $u_M$. Clearly, $u_M$ is a child of $v_M$ since $u$ is a child of $v$ and $u_M\neq v_M$. Moreover, $r_{u_M}$ must be the child of $v$ in the expanded $T_M$. Since $v'$ is an ancestor of $v$, that means its depth in the expanded $T_M$ satisfies $d(v') < d(r_{u_M})$. Since this defines the label of $v'$, $\ell(v) \leq d(r_{u_M}) - 1$.

Consider any leaf $l\in T_{u_M}$. Its label is the depth of some vertex $w\in T_{u_M}$, which must be deeper than the root $r_{u_M}$. Thus $\ell(l) \geq d(r_{u_M}) \geq \ell(v) + 1$, then implying all of $T_{u_M}$ is in $C$. Additionally, since only one $l\in T_{u_M}$ is labeled by $r_{u_M}$, i.e. $\ell(l) = d(r_{u_M})$, it must have a unique (and smallest) label out of all vertices in $P_{u_M}$. It turns out this $\ell(l)$ will be our $j$.

Now, consider $l$'s placement in the original tree $T$. It is on a heavy path containing $u$, and it has the smallest label on the heavy path. Vertex $u$ is also directly adjacent to $v$. Thus in level $i$ when we consider the tree induced on $T^i$, it must be in the same component as $u$ since the path from $l$ to $u$ is contained in $T^i$. Additionally, since $T$ is a tree, the only vertex in $N(C)$ above the component itself is $v$. Therefore, all of $C$ is a descendant of $v$, and thus they must also be descendants of $u$. Thus, in $T_M$, they must have been in meta vertices at the depth of or below $u_M$. By a similar logic as before, their label must be strictly greater than $l$'s label. This implies that $\ell(l) = j$, and it is the only vertex in $T_{u_M}$ with such a label.

\paragraph{Case 2} $u$ is the parent of $v$ and they are not on the same heavy path. Using $u_M$ and $v_M$ as before and with the same logic as before but with reversed roles, we find that $\ell(v) > \ell(u)$ so $\ell(u) < i$. This contradicts that $u$ is in $T^i$, and thus this case is impossible. Note that this also shows that $C$ cannot contain a vertex in an ancestor $u_M$ of $v_M$ in $T_M$, because by connectivity, this would imply that there is some $u$ that is the parent of the root of $T_{v_M}$. A similar argument will hold.

\paragraph{Case 3} $u$ and $v$ are on the same heavy path corresponding to meta vertex $u_M$ whose binarized path $T_{u_M}$ has root $r_{u_M}$. Let $P\subseteq P_{u_M}$ be the connected subpath of this heavy path that contains $u$ if we remove all vertices of level $i-1$ or lower. Obviously, $P= C \cap P_{u_M} \subseteq C$. Assume for contradiction that two vertices $p,p'\in P$ have level $i$, so $\ell(p)=\ell(p')=i$.  Let $a = lca_M(p,p')$ be the least common ancestor of $p$ and $p'$ in $T_{u_M}$. Because $a$ is the least common ancestor, one of its children must contain $p$ and the other $p'$ in its subtree. Without loss of generality, assume $p$ is a descendant of the left child and $p'$ is a descendant of the right child.

Let $w$ and $w'$ be the (possibly internal) vertices of $T_{u_M}$ such that $p$ and $p'$ are labeled with their expanded meta tree depths in step 3 respectively ($\ell(p) = d(w)$ and $\ell(p') = d(w')$).  This means $i = d(w) = d(w')$. Also, notice that $w$ is an ancestor of $p$ and $w'$ is an ancestor of $p'$, so they are on the respective paths from $p$ and $p'$ to the root of $T_{u_M}$, call it $r_{u_M}$. Additionally, $a$ is on both paths, and specifically the paths must merge at $a$. Since $w$ and $w'$ only pass their depth to exactly one leaf of $T_{u_M}$ each, they must be distinct. Therefore, they cannot be at depth $d_{u_M}(a)$ or lower, else they would both be $a$. Thus, $w$ and $w'$ are strict descendants of $a$, so $d(w), d(w') > d(a)$, meaning that $\ell(p), \ell(p') > d(a)$.

Let $\bar{p}$ be the leftmost leaf descendant of the right child of $a$. Then $\bar{p}$ was labeled by $a$, thus $\ell(\bar{p}) = d(a)$. Therefore, $\ell(p), \ell(p') > \ell(\bar{p})$. Since $p$ and $p'$ are at level $i$, $\ell(\bar{p}) < i$. Additionally, note that by the structure of $T_{u_M}$, $\bar{p}$ is necessarily located between $p$ and $p'$ in the leaf set, thus it must come between them in the pre-order traversal of $T_{u_M}$ by Observation~\ref{obs:ancestors}. Since $T_{u_M}$ is the binarized path of $P_{u_M}$, this also implies $\bar{p}$ comes between $p$ and $p'$ in $P_{u_M}$, and by extension $P$. Note, however, that $P\subseteq C\subseteq T^i$, and thus its vertices must all have level $i$ or greater. This contradicts that $\ell(\bar{p}) < i$. Thus, it must be the case that there is at most vertex in $P$ that has level $i$.

Finally, all that needs to be shown is that for every vertex $c\in C\setminus P$, $\ell(c) > i$. For any such $c$, let $c^*$ be its lowest ancestor in $P_{u_M}$. Since $\ell(c^*) \geq i$, we can make the same argument as Case 1 to show that $c$ must have label $i+1$ or larger.
\end{proof}

\begin{proof}[Proof of Lemma~\ref{lem:treedecomp}]
By combining Lemmas~\ref{lem:orient},~\ref{lem:heavylight},~\ref{lem:expandverts}, and~\ref{lem:decompcorrectness}, we have our result.
\end{proof}

\section{Missing Proofs from Section~\ref{sec:singleton}}

\begin{proof}[Proof of Lemma~\ref{lem:ldrtime}]
Consider $\bag(v, t)$ for some $v \in V$, $t \in [n^3]$. Because $v \in \bag(v, t)$, this bag cannot be empty. Let $u$ be the vertex with the smallest label among all vertices from $\bag(v, t)$. Assume for contradiction that there exists another vertex $u' \in \bag(v, t)$ such that $\ell(u) = \ell(u')$. Observe, that there exists only one path between $u$ and $u'$ in tree $T$ and all vertices of this path must be contained in $\bag(v, t)$ since $\bag(v, t)$ always forms a connected component when viewed as a subtree of $T$, c.f. Definition~\ref{def:bag}.
From the properties of the low depth decomposition, we get that there is a vertex $z$ on the path with label smaller than both  $\ell(u)$ and $\ell(u')$ which contradicts with the choice of $u$. This also proves the uniqueness of the leader.
Let us now argue about $\ldrtime(v)$. First, for every $v \in V$ we have that $bag(v, 0) = v$ as there is no edge with weight smaller than $1$. This already proves two first properties. Second, it holds: $\bag(v, 0) \subseteq \bag(v, 1), \ldots, \bag(v, n^3)$ as a bag defined for larger time can expand more edges. Therefore, if $v$ is the leader of $\bag(v, t)$ for some $0 \le t \le n^3$, it has to be the leader of any subset of this bag.
\end{proof}

\begin{proof}[Proof of Lemma~\ref{lem:two-edges-decom}]
Let $v$ be a vertex of this component with lowest value $\ell(v)$. By the construction of the low depth decomposition, the component $C_{i}$ consists of some connected part $P_{1}$ of binarized path $P$ to which $v$ belongs and all binarized paths incident to the part $F$ that contain vertices with larger values $\ell$, denote this set $P_{2}$. Note, that each vertex $v \in P_{2}$ has no edged to the part of graph $V \setminus T^{i}$.  Thus, only vertices from part $P_{1}$ can be connected with vertices of smaller labels than $i$. However, since $P_{1}$ forms a path, and we consider only tree edges, then there can be at most two such edges. 

The above proof instructs also how to compute these two edges. First, the vertex $v$ can determined in $O(1 / \epsilon)$ rounds, since it requires computing max function only over labels of vertices belonging to $C_{i}$. Assume, that for each vertex of a tree, we store not only its value $\ell(v)$, but also its position and the length of the binarized path to which it belongs. Then the first vertex to the left with label smaller than $v$ and the first vertex to the right with label smaller than $v$ on $v$'s binarized path can be computed in constant time in local memory of a single machine, since their positions in the binarized path are functions of only the length of the path and the position of $v$. Assumed, that in the global memory all binarized paths are stored, then the corresponding edges can be found in $O(1)$ queries to the global memory.
\end{proof}

\end{document}